\documentclass[11pt]{article}
\usepackage{times}

\usepackage{changepage}   
\usepackage{sidecap}
\usepackage{frcursive}
\usepackage{setspace}
\usepackage{caption}
\usepackage{titlesec}

\titleclass{\subsubsubsection}{straight}[\subsection]

\newcounter{subsubsubsection}[subsubsection]
\renewcommand\thesubsubsubsection{\thesubsubsection.\arabic{subsubsubsection}}

\titleformat{\subsubsubsection}
  {\normalfont\normalsize\bfseries}{\thesubsubsubsection}{1em}{}
\titlespacing*{\subsubsubsection}
{0pt}{3.25ex plus 1ex minus .2ex}{1.5ex plus .2ex}

\makeatletter
\renewcommand\paragraph{\@startsection{paragraph}{5}{\z@}%
  {3.25ex \@plus1ex \@minus.2ex}%
  {-1em}%
  {\normalfont\normalsize\bfseries}}
\renewcommand\subparagraph{\@startsection{subparagraph}{6}{\parindent}%
  {3.25ex \@plus1ex \@minus .2ex}%
  {-1em}%
  {\normalfont\normalsize\bfseries}}
\def\toclevel@subsubsubsection{4}
\def\toclevel@paragraph{5}
\def\toclevel@paragraph{6}
\def\l@subsubsubsection{\@dottedtocline{4}{9em}{4em}}
\def\l@paragraph{\@dottedtocline{5}{10em}{5em}}
\def\l@subparagraph{\@dottedtocline{6}{14em}{6em}}
\makeatother

\usepackage{tikz}

\usetikzlibrary{arrows}
\usetikzlibrary{shapes}

\usepackage{kbordermatrix,blkarray}

\renewcommand{\kbldelim}{(}
\renewcommand{\kbrdelim}{)}

\usepackage{wrapfig}
\usepackage{colortbl}
\usepackage{etex}
\usepackage{bbding}
\usepackage{dingbat}
\usepackage[ruled]{algorithm2e}
\usepackage{amsmath}
\usepackage{algorithmic}
\usepackage{amsfonts}
\usepackage{amsmath}
\usepackage{amssymb}
\usepackage{latexsym}
\usepackage{graphicx}
\usepackage{color}
\usepackage[normalem]{ulem}
\usepackage{balance}
\usepackage{textcomp}
\usepackage{nicefrac}
\usepackage{wasysym}
\usepackage{enumitem}
\usepackage{gensymb}
\usepackage{ctable}
\usepackage{url}
\usepackage{longtable}
\usepackage{citesort}

\usepackage{pstricks,pst-node,pst-plot,pst-coil,pst-text,multido,pst-3d,pst-tree,pst-grad,graphics,graphicx}

\usepackage{frcursive}

\usepackage{pifont,epsfig,rotating}

\usepackage{etoolbox}

\DeclareMathAlphabet{\mathpzc}{OT1}{pzc}{m}{it}

\usepackage{etaremune}

\usepackage{mathtools}

\definecolor{dgreyblue}{rgb}{0.26,0.3,0.46}             

\newcommand{\cA}{\mathcal{A}}
\newcommand{\cB}{\mathcal{B}}
\newcommand{\cC}{\mathcal{C}}
\newcommand{\cD}{\mathcal{D}}

\newcommand{\cZ}{\mathcal{Z}}

\renewcommand{\text}[1]{\hbox{\rm \ #1\ \/}}

\newcommand{\be}[1]{\begin{equation}\label{#1}}
\newcommand{\ee}{\end{equation}}
\newcommand{\beqn}{\begin{eqnarray*}}
\newcommand{\eeqn}{\end{eqnarray*}}
\newcommand{\beq}{\begin{eqnarray}}
\newcommand{\eeq}{\end{eqnarray}}
\newcommand{\ben}{\begin{enumerate}}
\newcommand{\een}{\end{enumerate}}
\newcommand{\bi}{\begin{itemize}}
\newcommand{\ei}{\end{itemize}}
\newcommand{\eps}{\varepsilon}

\newcommand{\IE}{{\em i.e.}\xspace}
\newcommand{\tx}{^{\mathrm{th}}}

\newtheorem{theorem}{Theorem}
\newtheorem{remark}{Remark}
\newtheorem{lemma}[theorem]{Lemma}
\newtheorem{corollary}[theorem]{Corollary}
\newtheorem{definition}[theorem]{Definition}

\newenvironment{proof}{{\noindent\bf Proof.\ }}{\hfill{\Pisymbol{pzd}{113}}\vspace{0.1in}}
\newenvironment{proof-sketch}{{\noindent\bf Sketch of Proof.\ }}{\hfill{\Pisymbol{pzd}{113}}\vspace{0.1in}}

\newcommand{\NP}{\mathsf{NP}}

\newcommand{\cP}{\mathcal{P}}
\newcommand{\EA}{{\em et al.}\xspace}
\newcommand{\TB}{\vspace{-0.1ex}}\newcommand{\TiE}{\setlength{\itemsep}{-1ex}}

\newcommand{\comment}[1]{}

\newcommand{\EG}{{\it e.g.}\xspace}
\newcommand{\FI}[1]{Fig.~\ref{#1}\xspace}

\newcommand{\Gr}{{\mathfrak{G}}}
\newcommand{\Ed}{{\mathfrak{E}}}
\newcommand{\Ve}{{\mathfrak{V}}}

\newcommand{\iif}{{\bf{if}}}
\newcommand{\tthen}{{\bf{then}}}

\newcommand{\ffor}{{\bf{for}}}

\newcommand{\rrepeat}{{\bf{repeat}}}
\newcommand{\ddo}{{\bf{do}}}

\newcommand{\cQ}{\mathcal{Q}}

\newcommand{\opt}{ {\mathsf{OPT}} }

\newcommand{\eqdef}{\stackrel{\mathrm{def}}{=}}

\newcommand{\mwvp}{{\sc Min}-{\sc wvp}}
\definecolor{columbiablue}{rgb}{0.61, 0.87, 1.0}

\newcommand{\aalpha}{{\mathrm{PartyA}}}
\newcommand{\bbeta}{{\mathrm{PartyB}}}
\newcommand{\eeta}{{\mathrm{Pop}}}

\allowdisplaybreaks

\title{Alleviating partisan gerrymandering: can math and computers help to eliminate wasted votes?
\thanks{The raw data reported in this paper are 
archived at http://www.cs.uic.edu/{\raise.17ex\hbox{$\scriptstyle\sim$}}dasgupta/gerrymander/index.html.
The source code of our implemented program will be eventually made freely available from the same website
after formal publication. All the authors declare \emph{no} competing interests.}
}

\author{
Tanima Chatterjee\thanks{Partially supported by NSF grant IIS-1160995.}
\qquad 
Bhaskar DasGupta$^\dagger$
\qquad 
Laura Palmieri
\and
Zainab Al-Qurashi
\qquad 
Anastasios Sidiropoulos\thanks{Supported by NSF CAREER award 1453472 and NSF grant CCF-1423230.}
\\
Department of Computer Science 
\\
University of Illinois at Chicago
\\
Chicago, IL 60607, USA
\\
Emails: {\sf \{tchatt2,bdasgup,lpalmi3,zalqur2,sidiropo\}@uic.edu}
}

\date{(Preliminary draft, subject to further revision, \today)}

\begin{document} 

\maketitle 

\begin{abstract}
Partisan gerrymandering is a major cause for voter disenfranchisement in United States.
However, convincing US courts to adopt specific measures to quantify gerrymandering has been of limited success to date. 
Recently, McGhee in~\cite{m14} and Stephanopoulos and McGhee in~\cite{sm15}
introduced a new and precise measure of partisan gerrymandering via the so-called 
''efficiency gap'' that computes the absolutes difference of wasted votes 
between two political parties in a two-party system. 
Quite importantly from a \emph{legal point of view}, this measure was found \emph{legally convincing} enough in a US appeals court in 
a case that claims that the legislative map of the state of Wisconsin was gerrymandered;  
the case is now pending in US Supreme Court (Gill v.\ Whitford, US Supreme Court docket no $16$-$1161$, decision pending).
In this article, we show the following: 
\begin{enumerate}[label=$\triangleright$]
\item
We provide interesting mathematical and computational complexity properties of 
the problem of minimizing the efficiency gap measure.
To the best of our knowledge, these are the \emph{first non-trivial} 
theoretical and algorithmic analyses of this measure of gerrymandering. 
\item
We provide a simple and fast algorithm that 
can ``un-gerrymander'' 
the district maps for the states of Texas, Virginia, Wisconsin 
and Pennsylvania
by bring their efficiency gaps to acceptable levels from the current unacceptable levels. 
Our work thus shows that, notwithstanding the general worst-case approximation hardness of the efficiency gap measure
as shown by us,
finding district maps with acceptable levels of efficiency gaps 
is a \emph{computationally tractable problem from a practical point of view}.
Based on these empirical results, we also provide some interesting insights into three 
practical issues related the efficiency gap measure.
\end{enumerate}
We believe that, should the US Supreme Court uphold the decision of lower courts, 
our research work and software will provide a crucial \emph{supporting hand} to remove partisan gerrymandering.
\end{abstract}


\section{Introduction and Motivation}

Gerrymandering, namely creation of district plans with highly asymmetric electoral outcomes
to disenfranchise voters, has continued to be a curse to fairness of electoral systems in USA for a long time in spite 
of general public disdain for it. Even though 
the US Supreme Court ruled in $1986$~\cite{t1}
that gerrymandering is \emph{justiciable}, 
they could not agree on an effective way of estimating it. Indeed,
a huge impediment to removing gerrymandering lies in 
formulating an effective and precise measure for it that will be acceptable in courts.

In $2006$, the US Supreme Court opined~\cite{t2} 
that a measure of \emph{partisan symmetry} may be a helpful
tool to understand and remedy gerrymandering.
Partisan symmetry is a standard for 
defining partisan gerrymandering that 
involves the computation of counterfactuals
typically under the assumption of uniform swings. 
To illustrate lack of partisan symmetry 
consider a two-party voting district and suppose that Party~A wins 
by getting $60\%$ of total votes and $70\%$ of total seats. 
In such a case, a partisan symmetry standard would hold if 
Party~B would also win $70\%$ of the seats had it won $60\%$ of the votes
in a \emph{hypothetical} election.
%
%
Two frequent indicators cited for lack of partisan symmetry are 
\emph{cracking}, namely 
dividing supporters of a specific party between two
or more districts when they could be a majority in a single district, 
and \emph{packing}, namely 
filling a district with more and more supporters of a specific party 
as long as this does not make this specific party the winner in that district.

There have been many theoretical and empirical attempts at remedying the lack of partisan symmetry
by ``quantifying'' gerrymandering and devising redistricting methods to optimize such 
quantifications using well-known notions such as 
\emph{compactness} and \emph{symmetry}~\cite{NGCH90,CR15,C85,wcl16,ND78,J94,Alt02,AG94}.
Since it is often simply \emph{not} possible to go over every possible redistricting map to optimize the gerrymandering 
measure due to rapid \emph{combinatorial explosion}, researchers such as~\cite{LWCW16,TL67,CR15} have also investigated 
designing efficient \emph{algorithmic approach} for this purpose.
In particular, a popular gerrymandering measure 
in the literature is \emph{symmetry}, which attempts to quantify the discrepancy between 
the share of votes and the share of seats of a party~\cite{ND78,J94,Alt02,AG94}.
In spite of such efforts, their success in convincing courts to adopt one or more of these measures has been 
unfortunately somewhat limited to date.

Recently, researchers Stephanopoulos and McGhee
in two papers~\cite{m14,sm15} have introduced a new gerrymandering measure called the 
``efficiency gap''. Informally speaking, the efficiency gap measure attempts to \emph{minimize} the 
\emph{absolute difference} of total wasted
votes between the parties in a two-party electoral system. 
This measure is very promising in several aspects. Firstly, it 
provides a mathematically precise measure of gerrymandering with many desirable properties. Equally importantly, 
at least from a legal point of view, this measure was found legally convincing in a US appeals court in 
a case that claims that the legislative map of the state of Wisconsin is gerrymandered;  
the case is now pending in US Supreme Court~\cite{t3}.

\subsection{Informal Overview of Our Contribution and Its Significance}

Redistricting based on minimizing the efficiency gap measure however requires one to find a solution to an 
combinatorial optimization problem.
To this effect, the contribution of this article is as follows:
\begin{enumerate}[label=$\triangleright$]
\item
As a necessary first step towards investigating the efficiency gap measure, 
in Section~\ref{sec-form}
we first formalize the optimization problem that corresponds 
to minimizing the efficiency gap measure.
\item
Subsequently, in Section~\ref{sec-math-prop}
we study the mathematical properties of the formalized version of the measure.
Specifically, 
Lemma~\ref{lem-abs} and Corollary~\ref{cor-abs}
show that the efficiency gap measure attains only a \emph{finite discrete set} of \emph{rational} values; 
these properties are of considerable importance in understanding 
the sensitivity of the measure and in designing efficient algorithms for computing this measure.
\item
Next, in Sections~\ref{sec-complexity}~and~\ref{sec-approx} we investigate computational complexity and algorithm design issues 
of redistricting based on the efficiency gap measure.
Although Theorem~\ref{thm-hardness}
shows that in theory one can construct \emph{artificial pathological examples} for which designing efficient algorithms
is provably hard, Theorem~\ref{app1} and Theorem~\ref{bhubhu} 
provide justification as to why the result in Theorem~\ref{thm-hardness}
is overly pessimistic for real data that do not necessarily correspond to these pathological examples. 
For example, 
assuming that the districts are \emph{geometrically compact} (\emph{$y$-convex} in our terminology), 
Theorem~\ref{bhubhu} shows how to find a district map \emph{efficiently in polynomial time} that \emph{minimizes} the 
efficiency gap.
\item
Finally, to show that it is indeed possible \emph{in practice} to solve the problem of minimization of the efficiency gap, 
in Section~\ref{sec-empirical} 
we design a \emph{fast randomized} algorithm based on the \emph{local search paradigm in combinatorial optimization} for this problem 
(cf.\ \FI{alg1}).
Our resulting software was tested on four electoral data 
for the $2012$ election of the (federal) house of representatives for  
the US states of Wisconsin~\cite{web-map-wi,web-wi}, Texas~\cite{web-map-tx,web-tx}, 
Virginia~\cite{web-va,web-map-va} and Pennsylvania~\cite{web-pa,web-map-pa}.
\emph{The results computed by our fast algorithm are truly outstanding: 
the final efficiency gap was lowered to $3.80\%$, $3.33\%$, $3.61\%$ and $8.64\%$
from $14.76\%$, $4.09\%$, $22.25\%$ and $23.80\%$ 
for Wisconsin, Texas, Virginia and Pennsylvania, respectively, in a small amount of time}.
Our empirical results clearly show that it is very much possible 
to design and implement a very fast algorithm 
that can ``un-gerrymander'' (with respect to the efficiency gap measure) 
the gerrymandered US house districts of four US states.

Based on these empirical results, we also provide some interesting insights into three 
practical issues related the efficiency gap measure, namely issues pertaining to 
\emph{seat gain vs.\ efficiency gap},
\emph{compactness vs.\ efficiency gap} and 
\emph{the naturalness of original gerrymandered districts}. 
\end{enumerate}
To the best of our knowledge, our results are \emph{first algorithmic analysis and implementation of minimization 
of the efficiency gap measure}.
Our results show that \emph{it is practically feasible to redraw district maps in a small amount of time to remove 
gerrymandering based on the efficiency gap measure}.
Thus, should the Supreme Court uphold the ruling of the lower court, 
our algorithm and its implementation will be a \emph{necessary} and \emph{valuable} asset to remove partisan gerrymandering.

\subsection{Beyond scientific curiosity: impact on US judicial system}

Beyond its scientific implications on the science of gerrymandering, 
we expect our algorithmic analysis and results to have a beneficial impact on the US judicial system
also. Some justices, whether at the Supreme Court level or in lower courts, 
seem to have a reluctance to taking mathematics, statistics and computing seriously~\cite{R03,Fa89},
For example, during the hearing in our previously cite most recent US Supreme Court case on gerrymandering~\cite{t3},
some justices opined that the math was \emph{unwieldy}, \emph{complicated}, and \emph{newfangled}. 
One justice called it \emph{baloney}, and 
Chief Justice John Roberts dismissed the attempts to quantify partisan gerrymandering by saying 

\begin{quote}
``It may be simply my educational background, but I can only describe it as \emph{sociological gobbledygook}.''
\end{quote}

Our theoretical and computational results show that the math, whether complicated or not (depending on one's background), 
\emph{can} in fact yield fast accurate
computational methods that can indeed be applied to un-gerrymander the currently gerrymandered maps. 

\subsection{Some Remarks and Explanations Regarding the Technical Content of This Paper}

To avoid any possible misgivings or confusions regarding the technical content of the paper as well as 
to help the reader towards understanding the remaining content of this article, we believe the following 
comments and explanations may be relevant.
\emph{We encourage the reader to read this section and explore the references mentioned therein before proceeding further}.
\begin{enumerate}[label=$\blacktriangleright$,leftmargin=*]
\item
We employ a randomized local-search heuristic for combinatorial optimization for our algorithm in \FI{alg1}. 
Our algorithmic paradigm is \emph{quite different} from 
\emph{Markov Chain Monte Carlo simulation}, \emph{simulated annealing approach}, \emph{Bayesian methods} 
and related similar other methods (\EG, no temperature parameter, no Gibbs sampling, no 
calculation of transition probabilities based on Markov chain properties, \emph{etc}.).
Thus, for example, our algorithmic paradigm and analysis for the efficiency gap measure 
is different and incomparable to that used by 
researchers for other different measure, such as
by Herschlag, Ravier and Mattingly~\cite{HRM17},
by Fifield \EA~\cite{FHIT15}
or by Cho and Liu~\cite{wcl16}.

For a detailed exposition of \emph{randomized algorithms} the reader is referred to excellent textbooks 
such as~\cite{MR95,AS16} and for a 
detailed exposition of \emph{local-search algorithmic paradigm} in combinatorial optimization 
the reader is referred to the excellent textbook~\cite{AL03}.
\item
While we do provide several non-trivial theoretical algorithmic results, we do
not provide any theoretical analysis of the randomized algorithm in \FI{alg1}. 
The justification for this is that, due to Theorem~\ref{thm-hardness} and Lemma~\ref{lemma-hardness},
\emph{no such non-trivial theoretical algorithmic complexity results exist} in general 
assuming P$\neq\NP$ for deterministic local-search algorithms or 
assuming RP$\neq\NP$ for randomized local-search algorithms. 
One can attribute this to the usual ``difference between theory and practice'' doctrine.

For readers unfamiliar with 
the complexity-theoretic assumptions P$\neq\NP$ and RP$\neq\NP$, these are \emph{core} complexity-theoretic 
assumptions that have been routinely used for decades in the field of algorithmic complexity analysis.
For example, starting with the famous Cook's theorem~\cite{C71} in $1971$ and Karp's subsequent
paper in $1972$~\cite{K72}, the P$\neq\NP$ assumption is the central assumption in structural complexity
theory and algorithmic complexity analysis. For a detailed technical coverage of the basic structural
complexity field, we refer the reader to the excellent textbook~\cite{BDG95}.
\item
In this article we use the data at the county level as opposed to using data at finer (more granular) level 
such as the ``Voting Tabulation District'' (VTD) level 
(VTDs are the smallest units in a state for which the election data are available).
The reason for this is as follows.
Note that our algorithmic approach already returns an efficiency gap of below $4\%$ for three states (namely, 
WI, TX and VA), and for PA it cuts down the current efficiency gap by a factor of about $3$ (\emph{cf}.\ Table~\ref{t2}). 
This, together with the observation in~\cite[pp. 886-888]{sm15} that the efficiency gap should \emph{not} be minimized 
to a very low value to avoid unintended consequences, 
shows that even just by using county-level data our algorithm can already output 
\emph{almost desirable} (if not truly desirable) values of the efficient gap measure and thus, by 
Occam's razor principle\footnote{Occam's razor principle~\cite{Ock} states that 
``\emph{Entia non sunt multiplicanda praeter necessitatem}'' (\IE, 
more things should not be used than are necessary). It is also known as rule of parsimony 
in biological context~\cite{F72}. Overfitting is an example of violation of this principle.} 
widely used in computer science, we should not be using more data at finer levels. 
In fact, using more data at a finer level may lead to what is popularly known as ``overfitting'' in the 
context of machine learning and elsewhere~\cite{BA02} that may hide its true performance on yet unexplored maps.
In this context, our suggestion to future algorithmic researchers in this direction is to use 
a minimal amount of data that is truly necessary to generate an acceptable solution.
\item
In this article we are not comparing our approaches empirically to those in existing literature 
such as in~\cite{HRM17,FHIT15,wcl16}. The reason for this is that, to the best of our knowledge,
there is currently \emph{no} other published work that gives a software to optimize the \emph{efficiency gap 
measure}. In fact, it would be \emph{grossly unfair to other existing approaches} if we compare our results with 
their results. For example, suppose we consider an optimal result using an approach from~\cite{wcl16}
and find that it gives an efficiency gap of $15\%$ whereas the approach in this article gives an efficiency gap
of $5\%$. However, it would be grossly unfair to say that, based on this comparison, our algorithm is 
better than the one in~\cite{wcl16} since the authors in~\cite{wcl16} \emph{never} intended to minimize the efficiency gap. 
Furthermore, even the two maps cannot be compared directly by geometric methods since \emph{no court} has so far 
established a firm and unequivocal \emph{ground truth} on gerrymandering by having a ruling of the following form: 

\begin{quote}
[court]: ``\sout{a district map is gerrymandered \emph{if and only if} such-and-such conditions are satisfied}''
\end{quote}

\noindent
(the line is crossed out above just to doubly clarify that such a ruling does not exist).

For certain scientific research problems, 
algorithmic comparisons are possible because of the existence of ground 
truths (also called ''gold standards'' or ``benchmarks''). 
For example, different algorithmic approaches for \emph{reverse engineering} causal relationships
between components of a \emph{biological cellular system} can be compared by evaluating 
how close the methods under investigation are in recovering known gold standard networks
using widely agreed upon metrics such as \emph{recall rates} or \emph{precision values}~\cite{DL16}.
Unfortunately, for gerrymandering this is not the case and, in our opinion, comparison of
algorithms for gerrymandering that optimize \emph{substantially different} objectives should be viewed with
a grain of salt.
\item
The main research goal of this paper is to minimize the efficiency gap measure \emph{exactly as introduced 
by Stephanopoulos and McGhee in}~\cite{m14,sm15}. 
However, should future researchers like to introduce additional computable constraints or objectives, 
such as compactness or respect of community boundaries, 
on top of our efficiency gap minimization algorithm, it is a conceptually easy task to modify 
our algorithm in \FI{alg1} for this purpose.
For example, to introduce compactness on top of minimization of the efficiency gap measure, the following 
two lines in \FI{alg1} 

\begin{quote}
if $\mathsf{Effgap}_{\kappa}(\cP,\cQ'_1,\dots,\cQ'_\kappa)<\mathsf{Effgap}_{\kappa}(\cP,\cQ_1,\dots,\cQ_\kappa)$ then
\end{quote}

should be changed to something like (changes are indicated in \textbf{bold}): 

\begin{quote}
if $\mathsf{Effgap}_{\kappa}(\cP,\cQ'_1,\dots,\cQ'_\kappa)<\mathsf{Effgap}_{\kappa}(\cP,\cQ_1,\dots,\cQ_\kappa)$ 
\\
\hspace*{1in} \textbf{and each of $\pmb{\cQ'_1,\dots,\cQ'_\kappa}$ are compact}
then
\end{quote}

and appropriate \emph{minor} changes can be made to other parts of the algorithm for consistency with this modification.
\end{enumerate}

\section{Formalization of the Optimization Problem to Minimize the Efficiency Gap Measure}
\label{sec-form}

\begin{figure}[htbp]
\centerline{\includegraphics[scale=1.5]{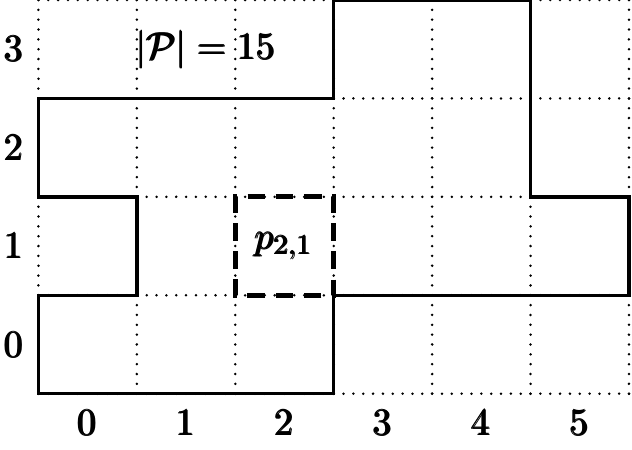}}
\caption{\label{ex1-fig}Input polygon $\cP$ of size $15$ placed on a grid of size $6\times 4$; the cell 
$p_{2,1}$ is shown.}
\end{figure}

Based on~\cite{m14,sm15}, we abstract our problem in the following manner.
We are given a rectilinear polygon $\cP$ without holes. Placing $\cP$ on a unit grid of size $m\times n$, 
we will identify an individual unit square (a ``\emph{cell}'') on the $i\tx$ row and $j\tx$ column 
in $\cP$ by $p_{i,j}$ for $0\leq i<m$ and $0\leq j<n$ (see~\FI{ex1-fig}). 
For each cell $p_{i,j}\in\cP$, we are given the following three integers: 
\begin{enumerate}[label=$\blacktriangleright$]
\item
an integer $\eeta_{i,j}\geq 0$ (the ``total population'' inside $p_{i,j}$), and 
\item
two integers $\aalpha_{i,j},\bbeta_{i,j}\geq 0$ 
(the total number of voters for Party~A and Party~B, respectively)
such that $\aalpha_{i,j}+\bbeta_{i,j}=\eeta_{i,j}$. 
\end{enumerate}
Let 
$|\cP|=
\left| \,
\left\{ 
p_{i,j} \,:\,p_{i,j}\in\cP 
\right\}
\, \right|
$
denote the ``size'' (number of cells) of $\cP$. 
For a rectilinear polygon $\cQ$ included in the interior of $\cP$ (\IE, a connected subset of the interior of $\cP$), 
we defined the following quantities: 
\begin{description}[leftmargin=0.2cm]
\item[Party affiliations in $\cQ$:] 
$\aalpha(\cQ)=\sum_{p_{i,j}\in\cQ}\aalpha_{i,j}$ 
and 
$\bbeta(\cQ)=\sum_{p_{i,j}\in\cQ}\bbeta_{i,j}$.
\item[Population of $\cQ$:] 
$\eeta(\cQ)=\aalpha(\cQ)+\bbeta(\cQ)$.
\item[Efficiency gap of $\cQ$:] 
$\mathsf{Effgap}(\cQ)=$
\[
\left\{
\begin{array}{r l}
\Big(\aalpha(\cQ) - \frac{1}{2}\eeta(\cQ) \Big) - \bbeta(\cQ)=2\aalpha(\cQ) - \frac{3}{2}\eeta(\cQ), 
      & \mbox{if $\aalpha(\cQ) \geq \frac{1}{2}\eeta(\cQ)$}
\\
[5pt]
\aalpha(\cQ) - \Big(\bbeta(\cQ) - \frac{1}{2}\eeta(\cQ) \Big)=2\aalpha(\cQ) - \frac{1}{2}\eeta(\cQ), 
      & \mbox{otherwise}
\end{array}
\right.
\]
Note that if 
$\aalpha(\cQ)=\bbeta(\cQ)=\frac{1}{2}\eeta(\cQ)$ then $\mathsf{Effgap}(\cQ)=-\bbeta(\cQ)$, \IE, 
\emph{in case of a tie, we assume Party~A is the winner}.
Also, note that 
$\mathsf{Effgap}(\cQ)=0$ if and only if 
either 
$\aalpha(\cQ)= {\eeta(\cQ)}/{4}$
or 
$\bbeta(\cQ)= {\eeta(\cQ)}/{4}$.
\end{description}
Our problem can now be defined as follows.
\medskip

\begin{adjustwidth}{0.6cm}{}
\fbox
{
\begin{minipage}{0.93\textwidth}
\begin{description}
\item[Problem name:]
$\kappa$-district  Minimum Wasted Vote Problem 
(\mwvp$_{\kappa}$).
\item[Input:]
a rectilinear polygon $\cP$ with 
$\eeta_{i,j},\aalpha_{i,j},\bbeta_{i,j}$ 
for every cell $p_{i,j}\in\cP$, 
and a positive integer $1<\kappa\leq |\cP|$.
\item[Definition:]
a $\kappa$-equipartition of $\cP$ is a partition of the interior of $\cP$ into exactly $\kappa$ rectilinear 
polygons, say $\cQ_1,\dots,\cQ_\kappa$, 
such that 
$\eeta(\cQ_1)=\dots=\eeta(\cQ_\kappa)$.
\item[Assumption:]
$\cP$ has at least one $\kappa$-equipartition.
\item[Valid solution:]
Any $\kappa$-equipartition $\cQ_1,\dots,\cQ_\kappa$ of $\cP$.
\item[Objective:]
\emph{minimize} 
the \emph{total} absolute efficiency gap\footnotemark 
$
\mathsf{Effgap}_{\kappa}(\cP,\cQ_1,\dots,\cQ_\kappa)
=
\left| \,
\sum_{j=1}^\kappa \mathsf{Effgap}(\cQ_j)
\,\right|
$.
\item[Notation:]
$
\opt_\kappa(\cP)\eqdef
\min
\left\{ \,\mathsf{Effgap}_{\kappa}(\cP,\cQ_1,\dots,\cQ_\kappa) 
\,\,|
\text{ $\cQ_1,\dots,\cQ_\kappa$ is a $\kappa$-equipartition of $\cP$} 
\right\}
$.
\end{description}
\end{minipage}
}
\end{adjustwidth}
\footnotetext{Note that 
our notation uses the absolute value for $\mathsf{Effgap}_{\kappa}(\cP,\cQ_1,\dots,\cQ_\kappa)$
but not for individual $\mathsf{Effgap}(\cQ_j)$'s.}

\section{Mathematical Properties of Efficiency Gaps: Set of Values Attainable by the Efficiency Gap Measure}
\label{sec-math-prop}

The following lemma sheds some light on the set of rational numbers that the total efficiency gap of a 
$\kappa$-equipartition can take. As an illustrative example, if we just partition the polygon $\cP$ 
into $\kappa=2$ regions, then 
$\mathsf{Effgap}_{2}(\cP,\cQ_1,\cQ_2)$
can only be one of the following $3$ possible values: 
\[
\left|\,2\,\aalpha(\cP)- \frac{1}{2}\,\eeta(\cP)\,\right|
\,\,\text{ or } \,\,
\left|\,2\,\aalpha(\cP)- \eeta(\cP)\,\right|
\,\,\text{ or } \,\,
\left|\,2\,\aalpha(\cP)- \frac{3}{2}\,\eeta(\cP)\,\right|
\]

\begin{lemma}\label{lem-abs}
$\,$
\begin{description}
\item[(\emph{a})]
For any 
$\kappa$-equipartition $\cQ_1,\dots,\cQ_\kappa$ of $\cP$, 
$\mathsf{Effgap}_{\kappa}(\cP,\cQ_1,\dots,\cQ_\kappa)$
always assumes one of the $\kappa+1$ values of the form 
$\left|\,2\,\aalpha(\cP)- \Big( z + \frac{\kappa}{2}  \Big) \frac{\eeta(\cP)}{\kappa}\,\right|$
for $z=0,1,\dots,\kappa$. 
\item[(\emph{b})]
If 
$
\mathsf{Effgap}_{\kappa}(\cP,\cQ_1,\dots,\cQ_\kappa)
=\left|\,2\,\aalpha(\cP)- \Big( z + \frac{\kappa}{2}  \Big) \frac{\eeta(\cP)}{\kappa}\,\right|$
for some 
$z\in\{0,1,\dots,\kappa\}$ 
and 
some $\kappa$-equipartition $\cQ_1,\dots,\cQ_\kappa$ of $\cP$, 
then 
$
\frac{\eeta(\cP)}{2\,\kappa}z \leq 
\aalpha(\cP)\leq 
\frac{\eeta(\cP)}{2\,\kappa}z
+
\frac{1}{2}\eeta(\cP)$.
\end{description}
\end{lemma}

\begin{proof}
To prove (\emph{a}), 
consider any 
$\kappa$-equipartition $\cQ_1,\dots,\cQ_\kappa$ of $\cP$ with 
$\eeta(\cQ_1)= \dots \eeta(\cQ_\kappa)=\frac{1}{\kappa}\eeta(\cP)$.
Note that for any $\cQ_j$ we have 
$
\mathsf{Effgap}(\cQ_j)=
2\aalpha(\cQ_j) - r_j\eeta(\cQ)
$
where
\[
r_j=\left\{
\begin{array}{r l}
\nicefrac{3}{2}, 
  & \mbox{if $\aalpha(\cQ_j) \geq {\eeta(\cQ_j)}/{(2\,\kappa)}$}
\\
\nicefrac{1}{2}, 
  & \mbox{otherwise}
\end{array}
\right.
\]
Letting $z$ be the number of $r_j$'s that are equal to $3/2$, it follows that 
\begin{multline*}
\mathsf{Effgap}_{\kappa}(\cP,\cQ_1,\dots,\cQ_\kappa)
=
\left| \,
\sum_{j=1}^\kappa \mathsf{Effgap}(\cQ_j)
\,\right|
=
\left|\,2\aalpha(\cP)- \Big( \frac{3}{2}z + \frac{1}{2}(\kappa-z)  \Big) \frac{\eeta(\cP)}{\kappa}\,\right|
\\
=
\left|\,2\,\aalpha(\cP)- \Big( z + \frac{\kappa}{2}  \Big) \frac{\eeta(\cP)}{\kappa}\,\right|
\end{multline*}
To prove (\emph{b}),
note that, since 
$0\leq\aalpha(\cQ_j) \leq \frac{\eeta(\cP)}{\kappa}$ for any $j$,  
we have 
$
\aalpha(\cP)=\sum_{j=1}^\kappa\aalpha(\cQ_j)
\geq
\sum_{j\,:\,r_j=3/2}\frac{\eeta(\cP)}{2\,\kappa}
=
\frac{\eeta(\cP)}{2\,\kappa}z
$
and 
$
\aalpha(\cP)=\sum_{j=1}^\kappa\aalpha(\cQ_j)
<
\sum_{j\,:\,r_j=3/2}\frac{\eeta(\cP)}{\kappa}
+
\sum_{j\,:\,r_j=1/2}\frac{\eeta(\cP)}{2\,\kappa}
=
\frac{\eeta(\cP)}{\kappa}z + 
\frac{\eeta(\cP)}{2\,\kappa} (\kappa-z)
=
\frac{\eeta(\cP)}{2\,\kappa}z
+
\frac{1}{2}\eeta(\cP)
$.
\end{proof}

\begin{corollary}\label{cor-abs}
Using the reverse triangle inequality of norms, 
the absolute difference between two successive values of 
$\mathsf{Effgap}_{\kappa}(\cP,\cQ_1,\dots,\cQ_\kappa)$
is given by 
\begin{multline*}
\left|
\,
\,
\left|
\,
2\, \aalpha(\cP) -  \left(  \frac{z}{\kappa} - \frac{1}{2} \right) \eeta(\cP)
\,
\right|
-
\left|
\,
2\, \aalpha(\cP) - \left(  \frac{z+1}{\kappa} - \frac{1}{2} \right) \eeta(\cP)
\,
\right|
\,
\,
\right|
\\
\leq 
\left|
\,
\,
\left(
\,
2\, \aalpha(\cP) -  \left(  \frac{z}{\kappa} - \frac{1}{2} \right) \eeta(\cP)
\,
\right)
\,-\,
\left(
\,
2\, \aalpha(\cP) - \left(  \frac{z+1}{\kappa} - \frac{1}{2} \right) \eeta(\cP)
\,
\right)
\,
\,
\right|
=
\frac{\eeta(\cP)}{\kappa} 
\end{multline*}
\end{corollary}

\begin{corollary}[see also~{\cite[p. 853]{sm15}}]
For any 
$\kappa$-equipartition $\cQ_1,\dots,\cQ_\kappa$ of $\cP$, 
consider the following quantities as defined in~{\em\cite{sm15}}: 
\begin{quote}
\begin{description}
\item[(Normalized) seat margin of Party~A:]
$
\frac{
\big| \left\{ 
\cQ_j \,:\,
\aalpha(\cQ_j) \geq \frac{1}{2}\eeta(\cQ_j)
\right\} \big|
}
{\kappa}
\,-\,
\frac{1}{2}
$.
\item[(Normalized) vote margin of Party~A:]
$
\frac{\aalpha(\cP)}{\eeta(\cP)}
- \frac{1}{2}
$.
\end{description}
\end{quote}
Then, we can write 
$\frac
 { 2\,\aalpha(\cP)- \Big( z + \frac{\kappa}{2}  \Big) \frac{\eeta(\cP)}{\kappa} }
 {\eeta(\cP)} 
$
as 
$
2 \left( \frac{\aalpha(\cP)}{\eeta(\cP)} - \frac{1}{2} \right)
-
\left( \frac{z}{\kappa}-\frac{1}{2} \right)
$,
and identifying $z$ with the quantity  
$\big| \left\{ \cQ_j \,:\, \aalpha(\cQ_j) \geq \frac{1}{2}\eeta(\cQ_j) \right\} \big|$
we get 
\[
\frac{\mathsf{Effgap}_{\kappa}(\cP,\cQ_1,\dots,\cQ_\kappa)}{\eeta(\cP)}
=
\big| \,
2 \times ( \text{vote margin of Party~A} )
- 
( \text{seat margin of Party~A} )
\,\big|
\]
\end{corollary}


\section{Approximation Hardness Result for \mwvp$_{\kappa}$}
\label{sec-complexity}

Recall that, for any $\rho\geq 1$, an approximation algorithm with an approximation ratio of $\rho$ (or, 
simply an $\rho$-approximation) is a polynomial-time solution of value at most $\rho$ times the value of 
an optimal solution~\cite{GJ79}.

\begin{theorem}\label{thm-hardness}
Assuming $\mathrm{P}\neq\NP$, 
for any rational constant $\eps\in(0,1)$, 
the \mwvp$_{\kappa}$ problem for a rectilinear polygon $\cP$ 
does not admit a $\rho$-approximation algorithm 
for any $\rho$ and all $2\leq\kappa\leq\eps |\cP|$.
\end{theorem}

\begin{remark}\label{rr1}
Since the PARTITION problem is not a strongly $\NP$-complete problem 
\emph{(\IE}, admits a pseudo-polynomial time solution\emph{)}, 
the approximation-hardness result in Theorem~\ref{thm-hardness} does not hold 
if the total population $\eeta(\cP)$ is polynomial in $|\cP|$ 
\emph{(\IE}, if $\eeta(\cP)=O(|\cP|^c)$ for some positive constant $c$\emph{)}.
Indeed, if 
$\eeta(\cP)$ is polynomial in $|\cP|$ then it is easy to design a polynomial-time exact 
solution via dynamic programming for those instances of 
\mwvp$_{\kappa}$ problem that appear in the proof of 
Theorem~\ref{thm-hardness}.
\end{remark}

\noindent
{\bf Proof of Theorem~\ref{thm-hardness}}. 
We reduce from the $\NP$-complete 
PARTITION 
problem~\cite{GJ79} which is defined as follows: 
{\em given a set of $n$ positive integers 
$\cA=\left\{ a_0,\dots,a_{n-1}\right\}$, decide if there exists a subset $\cA'\subset\cA$ 
such that 
$
\sum_{a_i\in\cA'} a_i
=
\sum_{a_j\notin\cA'} a_j 
$}.
Note that we can assume without loss of generality that $n$ is sufficiently large and each of 
$a_0,\dots,a_{n-1}$ is a multiple of \emph{any} fixed positive integer. 
For notational convenience, let $\Delta=\sum_{j=0}^{n-1} a_j$.

\begin{figure}[htbp]
\vspace*{-0.9in}
\centerline{\includegraphics[scale=1]{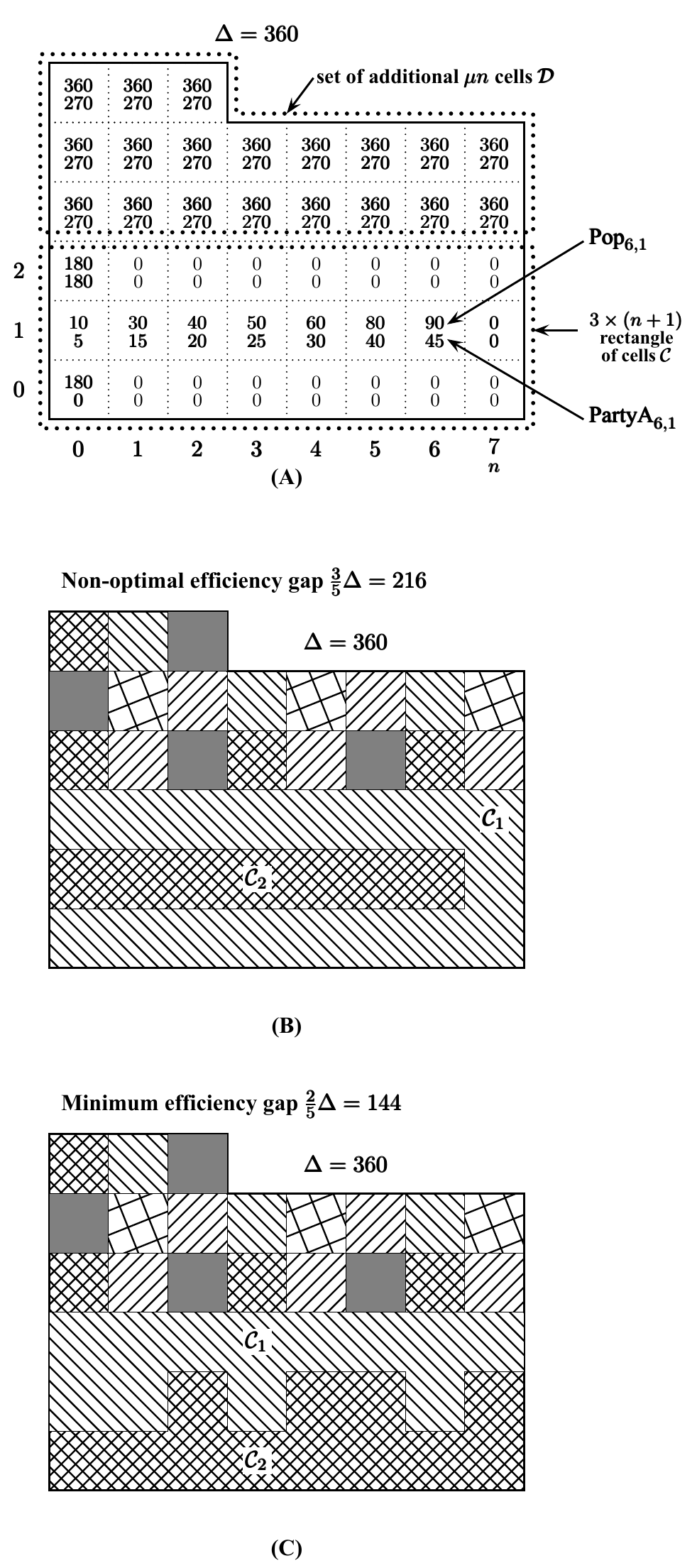}}
\caption{\label{ex2-fig}An illustration of the construction in the proof of Theorem~\ref{thm-hardness}
when the instance of the PARTITION problem is 
$\cA=\{10,30,40,50,60,80,90\}$.
}
\end{figure}

Let $\mu\geq 0$ be such that $\kappa=2+\mu n$ (we will later show that $\mu$ is at most the constant $\frac{6\eps}{1-\eps}$). 
Our rectilinear polygon $\cP$, as illustrated in~\FI{ex2-fig}(A), consists of a rectangle 
$\cC=\left\{ p_{i,j} \,|\, 0\leq i\leq n, \, 0\leq j\leq 2 \right\}$
of size $3\times (n+1)$ with 
additional $\mu n$ cells attached to it in any arbitrary manner to make the whole figure a connected polygon without holes.
For convenience, let 
$\cD=\left\{ p_{i,j} \,|\, p_{i,j}\notin \cC \right\}$
be the set of the additional $\mu n$ cells. 
The relevant numbers for each cell are as follows: 
\begin{gather*}
\eeta_{i,j}=
\left\{
\begin{array}{r l}
\nicefrac{\Delta}{2}, & \mbox{if $i=j=0$ or if $i=0,j=2$} 
\\
a_j, & \mbox{if $i=1$ and $j<n$}
\\
\Delta, & \mbox{if $p_{i,j}\in\cD$}
\\
0, & \mbox{otherwise}
\end{array}
\right.
\,\,\,\,\,
\,\,\,\,\,
\aalpha_{i,j}=
\left\{
\begin{array}{r l}
\nicefrac{\Delta}{2}, & \mbox{if $i=j=0$}
\\
\nicefrac{a_j}{2}, & \mbox{if $i=1$ and $j<n$}
\\
\nicefrac{3\Delta}{4}, & \mbox{if $p_{i,j}\in\cD$}
\\
0, & \mbox{otherwise}
\end{array}
\right.
\end{gather*}
First, we show how to select a rational constant $\mu$ such that any integer $\kappa$ in the range 
$[\,2,\eps |\cP|\,]$ can be realized.
Assume that 
$\kappa=\eps' |\cP|\in [\,2,\eps |\cP|\,]$ for some $\eps'$. 
Since 
$|\cP|=3(n+1)+\mu n$
the following calculations hold:
\begin{gather*}
\kappa=2+\mu n
=
\eps' |\cP|
=
\eps' (3(n+1)+\mu n)
\,\equiv\,
\mu = \frac{3\eps'(n+1)-2}{(1-\eps')n}
<
\frac{4\eps'}{1-\eps'}
<
\frac{4\eps}{1-\eps}
\end{gather*}

\noindent
\textbf{Claim~\ref{thm-hardness}.1}
$\mathsf{Effgap}(p_{i,j})=0$ 
for each $p_{i,j}\in\cD$, 
and moreover each $p_{i,j}\in\cD$ must be a separate partition by itself in any 
$\kappa$-equipartition of $\cP$. 

\smallskip

\begin{proof}
By straightforward calculation, 
$\mathsf{Effgap}(p_{i,j})=2\times \frac{3\Delta}{4}-\frac{3\Delta}{2}=0$.  
Since $\kappa=2+\mu n$ and 
$\eeta(\cP)=\sum_{p_{i,j}\in\cP} \eeta_{i,j}=\Delta + \Delta + \mu n \Delta = (2+\mu n)\Delta$, 
each partition in any 
$\kappa$-equipartition of $\cP$ must have a population of 
$\frac{\eeta(\cP)}{\kappa}=\Delta$ and thus 
each $p_{i,j}\in\cD$ of population $\Delta$ must be a separate partition by itself.
\end{proof}

Using Claim~\ref{thm-hardness}.1 we can simply ignore all $p_{i,j}\in\cD$ in the calculation of 
of efficiency gap of a valid solution of $\cP$ and it follows that 
the total efficiency gap of a $\kappa$-equipartition of $\cP$ 
is identical to that of a $2$-equipartition of $\cC$.
A proof of the theorem then follows provided we prove the following two claims.

\medskip
\begin{adjustwidth}{0.6cm}{}
\begin{description}
\item[\hspace*{0.2in}(soundness)]
If the PARTITION problem does not have a solution 
then 
$\opt_2(\cC)=\Delta$.
\item[(completeness)]
If the PARTITION problem has a solution 
then 
$\opt_2(\cC)=0$.
\end{description}
\end{adjustwidth}

\medskip
\noindent
\textbf{Proof of soundness} 
(refer to \FI{ex2-fig}(B))

\smallskip
Suppose that there exists a valid solution (\IE, a $2$-equipartition) 
$\cC_1,\cC_2$ of \mwvp$_2$ for $\cC$ with $p_{0,0}\in\cC_1,p_{0,2}\in\cC_2$, and let 
$\cA' =\left\{
a_j \,|\, p_{1,j}\in\cC_1 
\right\}
$.
Then,  
\[
\Delta 
= 
\frac{\eeta(\cC)}{2}
=
\eeta_{0,0} + \sum_{p_{1,j}\in\cC_1} \eeta_{1,j} 
=
\frac{\Delta}{2} + \sum_{a_j\in\cA'} a_j 
\,\equiv\,
\sum_{a_j\in\cA'} a_j = \frac{\Delta}{2} 
\]
and thus 
$\cA'$ is a valid solution of  
PARTITION, a contradiction!

Thus, assume that both 
$p_{0,0}$ and $p_{0,2}$
belong to the same partition, say $\cC_1$. Then, since 
$\eeta_{0,0}+\eeta_{0,4}=\Delta=\frac{\eeta(\cC)}{2}$, 
every $p_{1,j}$ must belong to $\cC_2$.
Moreover, every $p_{i,j}\in\cC$ with $\eta_{i,j}=0$ must belong to $\cC_1$ since 
otherwise $\cC_1$ will \emph{not} be a connected region. 
This provides 
$\eeta(\cC_1)=\eeta(\cC_2)=\Delta$, showing that  
$\cC_1,\cC_2$
is indeed a valid solution (\IE, a $2$-equipartition) of \mwvp$_2$ for $\cC$.
The total efficiency gap of this solution can be calculated as
\begin{multline*}
\mathsf{Effgap}_2(\cC,\cC_1,\cC_2)=
|\, \mathsf{Effgap}(\cC_1)+ \mathsf{Effgap}(\cC_2) \,|
\\
=
\left|\, 2 \,\aalpha(\cC_1) - \frac{3}{2}\eeta(\cC_1)  + 2 \,\aalpha(\cC_1) - \frac{3}{2}\eeta(\cC_1) \,\right|
=
\left|\, 2 \frac{\Delta}{2} - \frac{3}{2}\Delta    + 2 \frac{\Delta}{2} - \frac{3}{2}\Delta \,\right|
=
\Delta
\end{multline*}
%

\medskip
\noindent
\textbf{Proof of completeness} 
(refer to \FI{ex2-fig}(C))

\smallskip

Suppose 
that 
there is a valid solution of $\cA'\subset\cA$ of PARTITION and consider  
the two polygons 
\[
\cC_1=
\left\{ p_{2,j} \,|\, 0\leq j \leq n \right\}
\cup
\left\{ p_{1,j} \,|\, a_j\in\cA' \right\}, 
\,\,
\,\,
\,\,
\cC_2 = \cC\setminus\cC_1
\]
By straightforward calculation, it is easy to verify the following:
\begin{itemize}
\item 
$\eeta(\cC_1)=\sum_{a_j\in \cA'} a_j + \sum_{j=0}^n \eeta_{2,j}=\Delta$, 
$\eeta(\cC_2)=\sum_{a_j\notin \cA'} a_j + \sum_{j=0}^n \eeta_{2,j}=\Delta$, 
and thus $\cC_1,\cC_2$
is a valid solution (\IE, a $2$-equipartition) of \mwvp$_2$ for $\cC$.
\item 
$\mathsf{Effgap}_2(\cC,\cC_1,\cC_2)=0\opt_2(\cC)=0$ since 
\begin{multline*}
\mathsf{Effgap}_2(\cC,\cC_1,\cC_2)=
|\, \mathsf{Effgap}(\cC_1)+ \mathsf{Effgap}(\cC_2) \,|
\\
=
\left|\, 2 \,\aalpha(\cC_1) - \frac{3}{2}\eeta(\cC_1)  + 2 \,\aalpha(\cC_1) - \frac{1}{2}\eeta(\cC_1) \,\right|
\\
=
\left|\, 2 \,\left( \frac{\Delta}{2}+\frac{\Delta}{4} \right) 
    - \frac{3}{2}\Delta    + 2 \frac{\Delta}{4} - \frac{1}{2}\Delta \,\right|
=0
\end{multline*}
\end{itemize}
\hfill{\Pisymbol{pzd}{113}
\vspace{0.1in}

\section{Efficient Algorithms Under Realistic Assumptions}
\label{sec-approx}

Although Theorem~\ref{thm-hardness} seems to render the problem
\mwvp$_\kappa$ intractable in theory, our empirical results
show that the problem is computationally tractable in practice. 
This is because in real-life applications, many constraints in the theoretical 
formulation of \mwvp$_\kappa$ are often relaxed. For example: 
\begin{description}[leftmargin=10pt]
\item[(\emph{i}) Restricting district shapes:]
Individual partitions of the $\kappa$-equipartition of $\cP$ may be restricted in shape.
For example, 
$37$ states require their legislative districts to be reasonably compact and 
$18$ states require congressional districts to be compact~\cite{www}.
\item[(\emph{ii}) Variations in district populations:]
A partition $\cQ_1,\dots,\cQ_\kappa$ of $\cP$ is only \emph{approximately} $\kappa$-equipartition, \IE,
$\eeta(\cQ_1),\dots,\eeta(\cQ_\kappa)$ are approximately, but not exactly, equal to 
$\nicefrac{\eeta(\cP)}{\kappa}$.
For example, 
the usual federal standards 
require equal population \emph{as nearly as is practicable}
for congressional districts 
but allow more relaxed \emph{substantially equal} population 
(\EG, no more than $10\%$ deviation between the largest and smallest district)
for state and local legislative districts~\cite{www}. 
\item[(\emph{iii}) Bounding the efficiency gap measure away from zero:]
A $\kappa$-equipartition $\cQ_1,\dots,\cQ_\kappa$ of $\cP$ is a valid solution only if 
$\mathsf{Effgap}_{\kappa}(\cP,\cQ_1,\dots,\cQ_\kappa)\geq \eps\eeta(\cP)$ for some $0<\eps<1$.
Indeed, the authors that originally proposed the efficiency gap measure provided in~\cite[pp. 886-887]{sm15} 
several reasons for not requiring 
$\mathsf{Effgap}_{\kappa}(\cP,\cQ_1,\dots,\cQ_\kappa)/\eeta(\cP)$
to be either zero or too close to zero.
\end{description}
In this section, we explore algorithmic implications of these types of relaxations of constraints for 
\mwvp$_\kappa$.

\subsection{The Case of Two Stable and Approximately Equal Partitions}

This case considers constraints (\emph{ii}) and (\emph{iii}).
The following definition of ``near partitions''
formalizes the concept of variations in district populations.

\begin{definition}[Near partitions]
Let $\kappa\in \mathbb{N}$, and let $\cP$ be an instance of \mwvp$_\kappa$.
Let $\phi\eqdef\cQ_i,\ldots,\cQ_\kappa$ be such that $\cQ_1,\ldots,\cQ_\kappa$ is a partition of $\cP$, 
such that for each $i\in \{1,\ldots,\kappa\}$, we have 
\[
\left(\frac{1}{\kappa}-\delta\right) \cdot \eeta(\cP) \leq \eeta(\cQ_i) \leq \left(\frac{1}{\kappa}+\delta\right) \cdot \eeta(\cP),
\]
for some $\delta\geq 0$.
Then we say that $\phi$ is a \emph{$\delta$-near partition}.
\end{definition}

The next definition of ``stability''
formalizes the concept of bounding away from zero the efficiency gap of each partition.

\begin{definition}[Stability]
Let $\kappa \in \mathbb{N}$, let $\cP$ be an instance of \mwvp$_\kappa$, and let 
$\phi\eqdef\cQ_1,\ldots,\cQ_\kappa$ be a partition for $\cP$.
We say that $\phi$ is \emph{$\gamma$-stable}, for some $\gamma>0$, if for all $i\in \{1,\ldots,\kappa\}$, we have 
\[
\mathsf{Effgap}(\cQ_i) > \gamma \cdot \eeta(\cQ_i).
\]
\end{definition}

\begin{figure}
 \begin{center}
  \scalebox{0.8}{\includegraphics{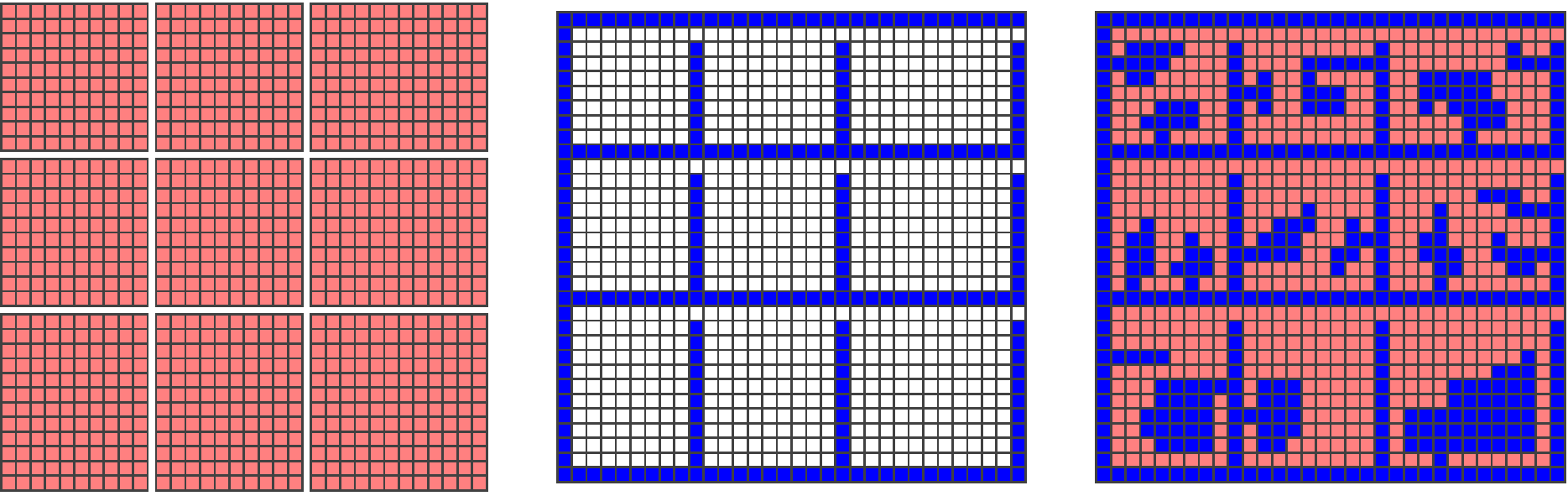}}
  \caption{The $10$-basic rectangles of a $32\times 32$ grid (left), the $10$-basic tree (middle), and a $10$-canonical solution (right).\label{fig:canonical}}
 \end{center}
\end{figure}

\begin{definition}[Canonical solution]
Let $\cP$ be an instance of \mwvp$_2$.
Let $t\in \mathbb{N}$.
Let $\cB$ be the partition of $\cP$ obtained as follows:
We partition $\cP$ into $\lfloor m/t\rfloor \cdot \rfloor n/t\rfloor$ rectangles, where each rectangle consists of the intersection of $t$ rows with $t$ columns, except possibly for the rectangles that are incident to the right-most and the bottom-most boundaries of $\cP$.
We refer to the rectangles in $\cB$ as \emph{$t$-basic} (see figure \ref{fig:canonical}).
Let $T$ be the set of cells consisting of the union of the left-most column of $\cP$, the top row of $\cP$, and for each $t$-basic rectangle $X$, the bottom row of $X$, and the right-most column of $X$, except for the cell that is next to the top cell of that column (see figure \ref{fig:canonical}).
We refer to $T$ as the \emph{$t$-basic tree}
For each $t$-basic rectangle $X$, we define its \emph{interior} to be the set of cells in $X$ that are at distance at least $2$ from $T$.
A solution $\phi\eqdef\cZ_1,\cZ_2$ of $\cP$ is called \emph{$t$-canonical} if it 
satisfies the following properties.
\begin{description}
\item{(1)}
$T\subseteq Z_1$.

\item{(2)}
Let $X$ be a $t$-basic rectangle, and let $X'$ be its interior.
For each $i\in \{1,2\}$, let $A_i$ be the set of connected components of $X'\cap \cZ_i$.
Then, for each $Y\in A_1$, there exists a unique cell in $T_1$ that is adjacent to both $Y$ and $T$.
Moreover, all other cells in $X\setminus (X'\cup T)$ are in $\cZ_2$ (see Figure \ref{fig:canonical}).
\end{description}
\end{definition}

\begin{lemma}\label{lem:canonical_existence}
Let $\cP$ be an instance of \mwvp$_2$.
Suppose that there are no empty cells and the maximum number of people in any cell of $\cP$ is $C$.
Let $\gamma>0$, and 
suppose that there exists a $\gamma$-stable solution $\phi\eqdef\cQ_1,\cQ_2$ for $P$.
Then, for any $\eps>0$, at least one of the following conditions hold:
\begin{description}
\item{(1)}
Either $\cQ_1$ or $\cQ_2$ is contained in some $1/\eps$-basic rectangle.
\item{(2)}
There exists a $\lceil 1/\eps \rceil$-canonical $\delta$-near solution $\phi'=\cQ'_1,\cQ'_2$ of $\cP$, 
for some $\delta=O(\eps C)$,
such that for all $i\in \{1,2\}$, we have 
$|\aalpha(\cQ_i)-\aalpha(\cQ_i')| \leq O(C\eps n+m)$,
and
$|\bbeta(\cQ_i)-\bbeta(\cQ_i')| \leq O(C\eps nm)$.
\end{description}
\end{lemma}

\begin{proof}
It suffices to show that if condition (1) does not hold, then condition (2) does.
We define a partition $\phi'\eqdef\cQ'_1,\cQ'_2$ of $\cP$ as follows.
We initialize $\cQ'_1$ to be empty.
Let $T$ be the $1/\eps$-basic tree.
We add $T_1$ to $\cQ'_1$.
For each $1/\eps$-basic rectangle $X$, let $X'$ be its interior.
For each $i\in \{1,2\}$, let $A_i$ be the set of connected components of $X'\cap \cQ_1$.
Since $\phi$ is a valid solution, we have that $\cQ_1$ is connected.
Since condition (1) does not hold, it follows that $\cQ_1$ intersects at least two $1/\eps$-basic rectangles.
Therefore, each component $W\in A_1$ must contain some cell $c_W$ on the boundary of $X'$.
By construction, $c_W$ must be incident to some cell $c_W'$ that is incident to $T$.
We add $c_W'$ to $\cQ'_1$.
Repeating this process for all basic rectangles, and for all components $W$ as above.
Finally, we define $\cQ'_2=\cP\setminus \cQ'_1$.
This completes the definition of the partition $\phi'\eqdef\cQ'_1,\cQ'_2$ of $\cP$.
It remains to show that this is the desired solution.

First, we need to show that $\phi'$ is a valid solution.
To that end, it suffices to show that both $\cQ'_1$ and $\cQ'_2$ are connected.
The fact that $\cQ'_1$ is connected follows directly from its construction.
To show that $\cQ'_2$ is connected we proceed by induction on the construction of $\cQ'_1$.
Initially, $\cQ'_1$ consists of just the cells in $T$, and thus its complement is clearly connected.
When we consider a component $W$, we add $W \cup \{c'_W\}$ to $\cQ'_1$.
Since we add only a single cell that is incident to both $W$ and $\cQ'_1$, it follows inductively that $\cQ'_1$ remains simply connected (that is, it does not contain any holes), and therefore its complement remains connected.
This concludes the proof that both $\cQ'_1$ and $\cQ'_2$ are connected, and therefore $\phi'$ is a valid solution.

The solutions $\phi$ and $\phi'$ can disagree only on cells that are not in the interior of any basic rectangle.
All these cells are contained in the union of $O(\eps m)$ rows and $(\eps n)$ columns.
Thus, the total number of voters in these cells is at most $O(C\eps nm)$.
It follows that for each $i\in \{1,2\}$, we have
$|\aalpha(\cQ_i)-\aalpha(\cQ_i')| \leq O(C\eps nm)$,
and
$|\bbeta(\cQ_i)-\bbeta(\cQ_i')| \leq O(C\eps nm)$.

Since there are no empty cells, we have $\eeta(\cP)\geq nm$.
It follows that for all $i\in \{1,2\}$, we have 
\[
 |\eeta(\cQ_i)-\eeta(\cQ'_i)| \leq O(C\eps nm) \leq C\cdot \eps \cdot \eeta(\cP).
\]
Thus $\phi'$ is $\delta$-near, for some $\delta=O(\eps C)$,
which concludes the proof.
\end{proof}

\begin{lemma}\label{lem:canonical_compute}
Let $\gamma>0$.
Let $\cP$ be an instance of \mwvp$_2$.
Suppose that there are no empty cells and the maximum number of people in any cell of $\cP$ is $C$.
Suppose that there exists a $\delta$-stable partition $\phi\eqdef\cQ_1,\cQ_2$ of $P$.
Then, for any fixed $\eps>0$, there exists an algorithm which given $\cP$ computes some $\delta$-near partition 
$\phi'\eqdef\cQ'_1,\cQ'_2$ of $\cP$,
for some $\delta=O(\eps C)$,
such that for all $i\in \{1,2\}$, we have 
$|\aalpha(\cQ_i)-\aalpha(\cQ_i')| \leq O(C\eps nm)$,
and
$|\bbeta(\cQ_i)-\bbeta(\cQ_i')| \leq O(C\eps nm)$,
in time $2^{O(1/\eps^2)} (nmC)^{O(1)}$.
\end{lemma}

\begin{proof}
We can check whether there exist a partition $\phi\eqdef\cQ_1, \cQ_2$ satisfying the conditions, and such that either 
$\cQ_1$ or $\cQ_1$ is contained in the interior of a single $1/\eps$-basic rectangle.
This can be done by trying all $1/\eps$-basic rectangles, and all possible subsets of the interior of each $1/\eps$-basic rectangle, in time $(n/\eps)(m/
\eps) 2^{O(1/\eps^2)} = n m 2^{O(1/\eps^2)}$.

It remains to consider the case where neither of $\cQ_1$ and $\cQ_2$ is contained in the interior of any $1/\eps$-basic rectangle.
It follows that condition (2) of Lemma \ref{lem:canonical_existence} holds.
That is, there exists some $\lceil 1/\eps \rceil$-canonical $\delta$-near solution $\phi'\eqdef\cQ'_1,\cQ'_2$ of $\cP$, 
for some $\delta=O(\eps C)$,
such that for all $i\in \{1,2\}$, we have 
$|\aalpha(\cQ_i)-\aalpha(\cQ_i')| \leq O(C\eps n+m)$,
and
$|\bbeta(\cQ_i)-\bbeta(\cQ_i')| \leq O(C\eps nm)$.
We can compute such a partition $\phi'$ via dynamic programming, as follows.
Let $I$ be the union of the interiors of all $1/\eps$-basic rectangles.
By the definition of a canonical partition, it suffices to compute $\cQ'_1\cap I$ and $\cQ'_2\cap I$.
Since $\cQ'_2 \cap I = I\setminus (\cQ'_1\cap I)$, it suffices to compute $\cQ'_1\cap I$.
Let $L_\aalpha=\aalpha(\cQ'_1\cap I)$,
and $L_\bbeta=\bbeta(\cQ'_1\cap I)$.
Clearly, $L_\aalpha,L_\bbeta\in \{0,\ldots,Cnm\}$.
Thus there are at most $O((nmC)^2)$ different values for the pair $(L_\aalpha,L_\bbeta)$.
We construct a dynamic programming table, containing one entry for each possible value for the pair $(L_\aalpha,L_\bbeta)$.
Initially, all entries of the table are unmarked, except for the entry that corresponds to the pair $(0,0)$.
We iteratively consider all $1/\eps$-basic rectangles.
When considering some $1/\eps$-basic rectangle $X$, with interior $X'$, we enumerate all possibilities for $Y=X'\cap \cQ'_1$.
There are $2^{O(1/\eps^2)}$ possibilities for $Y$.
For each such possibility, we update the dynamic programming table by marking the position $(i+\aalpha(Y), j+\bbeta(Y))$, if the position $(i,j)$ is already marked from the previous iteration.
The total running time is $2^{O(1/\eps^2)} (nmC)^{O(1)}$.
\end{proof}

\begin{theorem}\label{app1}
Let $\gamma>0$.
Let $\cP$ be an instance of \mwvp$_2$.
Suppose that there are no empty cells and the maximum number of people in any cell of $\cP$ is $C$.
Suppose that there exists a $\delta$-stable $\gamma$-near partition $\phi\eqdef \cQ_1,\cQ_2$ of $P$.
Then, for any fixed $\eps>0$, there exists an algorithm which given $\cP$ computes some $\delta$-near partition 
$\phi'\eqdef\cQ'_1,\cQ'_2$ of $\cP$,
for some $\delta=O(\eps C)$,
such that 
$\mathsf{Effgap}(\phi') \leq (1+\eps) \mathsf{Effgap}(\phi)$, in time $2^{O(1/\eps^2)} 2^{O(C/\delta)^2} (nmC)^{O(1)}$.
\end{theorem}

\begin{proof}
It follows from Lemma \ref{lem:canonical_compute} by setting $\eps=\min\{\eps, O(\delta / C)\}$.
\end{proof}

\subsection{The Case of Convex Shaped Partitions}

This case encompasses constraint (\emph{i}) since
convexity has been used in gerrymandering studies such as~\cite{azavea09} as a measure of 
compactness to examine how redistricting reshapes the geography of congressional districts. 
We recall that some $X\subseteq \mathbb{R}^2$ is called \emph{$y$-convex} if for every vertical line 
$\ell$, we have that $X\cap \ell$ is either empty, or a line segment.
We also say that a $\kappa$-partition 
${\cal Q}_1,\ldots,Q_\kappa$ of ${\cal P}$ is \emph{$y$-convex} if for all 
$i\in \{1,\ldots,\kappa\}$, ${\cal Q}_i$ is \emph{$y$-convex}.

\begin{theorem}\label{bhubhu}
Let ${\cal P}$ be a rectilinear polygon realized in the $m\times n$ grid, and let 
$N=\eeta({\cal P})$ be the total population on ${\cal P}$.
There exists an algorithm for computing a $y$-convex $\kappa$-equipartition of ${\cal P}$ of minimum efficiency gap, 
with running time $N^{O(\kappa)}$.
In particular, the running time is polynomial when the total population is polynomial and the total number of 
partitions is a constant.
\end{theorem}

\begin{proof}
Let ${\cal P}^* = {\cal Q}_1^*,\ldots,{\cal Q}_\kappa^*$ be a $y$-convex $\kappa$-equipartition of ${\cal P}$ of 
minimum efficiency gap.
For any $i\in \{1,\ldots,n\}$, let $C_i$ be the $i$-th column of ${\cal P}$.
We observe that for all $i\in \{1,\ldots,n\}$, and for all $j\in \{1,\ldots,\kappa\}$, 
we have that ${\cal Q}^*_j\cap C_i$ is either empty, or consists of a single rectangle of width $1$.
Let ${\cal C}_i$ be the set of all partitions of $C_i$ into exactly $\kappa$ (possibly empty) segments, 
each labeled with a unique integer in $\{1,\ldots,\kappa\}$.
We further define 
\[
a^*_{i,j} = \aalpha({\cal Q}_j^* \cap (C_1\cup \ldots \cup C_i)),
\]
and
\[
b^*_{i,j} = \bbeta({\cal Q}_j^* \cap (C_1\cup \ldots \cup C_i)).
\]

The algorithm proceeds via dynamic programming.
For each $i\in \{1,\ldots,n\}$, let $I_i=\mathbb{N}^{2\kappa} \times {\cal S}_i$.
Let $X_i=(a_{i,1},b_{i,1},\ldots,a_{i,\kappa},b_{i,\kappa},{\cal Z}_i)\in I_i$.
If $i=1$, then we say that $X_i$ is \emph{feasible} if for all $j\in \{1,\ldots,\kappa\}$, 
the unique set $Y\in {\cal Z}_1$ labeled $j$ satisfies 
\[
a_{1,j} = \aalpha(Y) \text{ and } b_{1,j} = \bbeta(Y).
\]
Otherwise, if $i>1$, we say that $X_i$ is feasible if the following holds:
There exists some feasible $X_{i-1}=(a_{i-1,1},b_{i-1,1},\ldots,a_{i-1,\kappa},b_{i-1,\kappa},{\cal Z}_{i-1})\in I_{i-1}$, 
such that for all $j\in \{1,\ldots,\kappa\}$, we have that the unique set $Y\in {\cal Z}_i$ labeled $j$ satisfies 
\[
a_{i,j} = a_{i-1,j} + \aalpha(Y) \text{ and } b_{i,j} = b_{i-1,j}+ \bbeta(Y).
\]
For each $i\in \{1,\ldots,n\}$ we inductively compute the set ${\cal F}_i$ of all feasible $X_i\in I_i$.
This can clearly be done in time $N^{O(\kappa)}$.
It is immediate that for all $i\in \{1,\ldots,n\}$, there exists some $X_i\in I_i$ that achieves efficiency gap equal to the 
restriction of ${\cal P}^*$ on the union of the first $i$ columns.
Thus, by induction on $i$, the algorithm computes a feasible solution with optimal efficiency gap.
\end{proof}

\section{Experimental Results for Real Data on Four Gerrymandered States}
\label{sec-empirical}

To show that it is indeed possible in practice to solve the problem of minimization of the efficiency gap, 
we design a \emph{fast randomized} algorithm based on the \emph{local search paradigm} for this problem.
Our algorithm starts with a given 
$\cQ_1,\dots,\cQ_\kappa$ partition of the input state $\cP$.
Note that 
$\cQ_1,\dots,\cQ_\kappa$ was only an \emph{approximate} equipartition
in the sense that the values 
$\eeta(\cQ_1),\dots,\eeta(\cQ_\kappa)$ are as close to each other as practically possible 
but need not be \emph{exactly} equal (\emph{cf}. US Supreme Court ruling in Karcher v.\ Daggett $1983$).
For designing alternate valid district plans, we therefore allow 
any partition $\cQ'_1,\dots,\cQ'_\kappa$ of $\cP$
such that, for every $j$, 
$\min_{1\leq i\leq\kappa}\left\{ \eeta(\cQ_i) \right\}\leq\eeta(\cQ'_j)\leq\max_{1\leq i\leq\kappa}\left\{ \eeta(\cQ_i) \right\}$.

\subsection{Input preprocessing}

We preprocess the input map to generate an undirected unweighted planar graph $\Gr=(\Ve,\Ed)$.
Each node in the graph corresponds to a planar subdivision of a county that is assigned 
to a district (or to an entire county if it is assigned to a district as a whole). Two
nodes are connected by an edge if and only if they share a border on the map. 
Each node $v\in \Ve$ has three corresponding numbers: 
$\aalpha(v)$ 
(total number of voters for Party~A), 
$\bbeta(v)=$
(total number of voters for Party~B), and 
$\eeta(v)=\aalpha(v)+\bbeta(v)$
(total population in $v$)\footnote{We ignore negligible ``third-party'' votes, \IE, 
votes for candidates other than the democratic and republican parties.}.
A district $\cQ$ is then a connected sub-graph of $\Gr$ 
with 
$\aalpha(\cQ)=\sum_{v\in\cQ}\aalpha(v)$ 
and 
$\bbeta(\cQ)=\sum_{v\in\cQ}\bbeta(v)$.  


\subsection{Availability and Format of Raw Data}

Link to all data files for the three states used in the paper are 
available in \url{http://www.cs.uic.edu/~dasgupta/gerrymander/index.html}.
Each data is an EXCEL spreadsheet. Explanations of various columns of the spreadsheet are as follows:
\begin{description}
\item[Column 1: District:]
This column identifies the district number of the county in column 3.
\item[Column 2: County\_id:]
Column 1 and column 2 together form an unique identifier for the counties in column 3. 
A county is identified by its County\_id (column 2) and the District (column 1) it belongs to. 
This was specifically needed to identify and differentiate the counties that belonged to more than one district.

The software considers the counties belonging to different districts as separate entities.
\item[Column 3: County:] 
This column contains the name of the county.
\item[Column 4: Republicans:] 
This contains the total number of votes in favor of the Republican party (GOP) in the county identified by column 1 and column 2.
\item[Column 4: Democrats:] 
This contains the total number of votes in favor of the Democratic party in the county identified by column 1 and column 2.
\item[Column 5: Neighbors:]
This column contains information about the ``neighboring counties'' of the given county. 
Neighboring counties represent the counties that share a boundary with the county identified by column 1 and column 2. 
Individual neighbors are separated by commas.
\end{description}

\subsection{The Local-search Heuristic}

Informally, our algorithm starts with the existing (possibly gerrymandered) districts and then repeatedly 
attempts to \emph{reassign} counties (or parts of counties) into neighboring districts. This was done on a \emph{semi-random} basis, 
and on average about $100$ iterations were carried out in each run. 
Each time a county (or a part of a county) was shifted, 
the efficiency gap was calculated to check if it was less than the prior efficiency gap.
Exact details of our algorithm are shown in \FI{alg1}.

\captionsetup{width=\textwidth}
\begin{table}[htbp]
\begin{tabular}{l }
\toprule
start with the current districts, say $\cQ_1,\dots,\cQ_{\kappa}$ 
\\
[3pt]
\rrepeat\ $\mu$ times 
    $(*$ $\mu$ was set to $100$ in actual run $*)$
\\
[3pt]
\hspace*{0.2in}
select a random $r$ from the set $\{0,1,\dots,k\}$ for some $0<k<|V|$ 
\hspace*{0.2in} 
$(*$ $k=20$ in actual run $*)$
\\
[3pt]
\hspace*{0.2in}
select $r$ nodes $v_1,\dots,v_r$ from $\Gr$ at random 
    \hspace*{0.2in}
    $(*$ Note that a node is a county or part of a county $*)$
\\
[3pt]
\hspace*{0.2in}
$\mathsf{counties\_done}\leftarrow\emptyset$
\\
[3pt]
\hspace*{0.2in}
\ffor\ each $v_i$ \ddo 
\\
[3pt]
\hspace*{0.4in}
\iif\ \emph{all} neighbors of $v_i$ do \emph{not} belong to the same district as $v_i$ \tthen 
\\
[3pt]
\hspace*{0.6in}
\iif\ $v_i\notin\mathsf{counties\_done}$ \tthen
\\
[3pt]
\hspace*{0.8in}
add $v_i$ to $\mathsf{counties\_done}$ 
\\
[3pt]
\hspace*{0.8in}
\ffor\ every neighbor $v_j$ of $v_i$ \ddo
\\
[3pt]
\hspace*{1.0in}
\iif\ assigning $v_i$ to the district of $v_j$ produces no district with disconnected parts \tthen 
\\
[3pt]
\hspace*{1.2in}
assign $v_i$ to the district of one of its neighbors
\\
[3pt]
\hspace*{1.2in}
recalculate new districts, say $\cQ'_1,\dots,\cQ'_{\kappa}$ 
\\
[3pt]
\hspace*{1.2in}
\iif\ 
$\min_{1\leq i\leq\kappa}\left\{ \eeta(\cQ_i) \right\}\leq\eeta(\cQ'_j)\leq\max_{1\leq i\leq\kappa}\left\{ \eeta(\cQ_i) \right\}$
for every $j$ \tthen
\\
[3pt]
\hspace*{1.4in}
\iif\ 
$\mathsf{Effgap}_{\kappa}(\cP,\cQ'_1,\dots,\cQ'_\kappa)<\mathsf{Effgap}_{\kappa}(\cP,\cQ_1,\dots,\cQ_\kappa)$
\tthen
\\
[3pt]
\hspace*{2in}
$\cQ_1 \leftarrow \cQ'_1$ ;  
$\cQ_2 \leftarrow \cQ'_2$ ;  
$\dots\dots$ ;
$\cQ_\kappa \leftarrow \cQ'_\kappa$
\\
[3pt]
\hspace*{1.4in}
\textbf{endif}
\\
\hspace*{1.2in}
\textbf{endif}
\\
\hspace*{1.0in}
\textbf{endif}
\\
\hspace*{0.8in}
\textbf{endfor}
\\
\hspace*{0.6in}
\textbf{endif}
\\
\hspace*{0.4in}
\textbf{endif}
\\
\hspace*{0.2in}
\textbf{endfor}
\\
\textbf{endrepeat}
\\
\bottomrule
\end{tabular}
\captionof{figure}{\label{alg1}A local search algorithm for computing efficiency gap. The algorithm 
was implemented using the PYTHON language.}
\end{table}

One cannot provide any theoretical analysis of the randomized algorithm in \FI{alg1} because no such 
analysis is possible (due to Theorem~\ref{thm-hardness}) as stated formally in the following lemma.

\begin{lemma}\label{lemma-hardness}
Assuming $P\neq\NP$ (respectively, $RP\neq\NP$), 
there exists no deterministic local-search algorithms
(respectively, randomized local-search algorithms) 
that reaches a solution with a finite approximation ratio 
in polynomial time starting at any non-optimal valid solution.
\end{lemma}

\begin{proof}
This follows from the proof of Theorem~\ref{thm-hardness} once observes that 
the specific instance of the \mwvp$_{\kappa}$ problem created in the reduction of the theorem 
has exactly one (trivial) non-optimal solution and every other valid solution is an optimal solution.
\end{proof}

\subsection{Results and Implications}

Our resulting software was tested on four electoral data 
for the $2012$ election of the (federal) house of representatives for  
the US states of Wisconsin~\cite{web-map-wi,web-wi}, Texas~\cite{web-map-tx,web-tx}, 
Virginia~\cite{web-va,web-map-va} and Pennsylvania~\cite{web-pa,web-map-pa}.
Some summary statistics for these data are shown in Table~\ref{ttt1}.
The results of running the local-search algorithm in \FI{alg1} on the four real data-sets are tabulated in Table~\ref{t2},
and the corresponding maps are shown in 
\FI{WImap-alg1}--\ref{PAmap-alg1}.
\emph{The results computed by our algorithm are truly outstanding: 
the final efficiency gap was lowered to $3.80\%$, $3.33\%$, $3.61\%$ and $8.64\%$
from $14.76\%$, $4.09\%$, $22.25\%$ and $23.80\%$ 
for Wisconsin, Texas, Virginia and Pennsylvania, respectively, in a small amount of time}.
Our empirical results clearly show that it is very much possible 
to design and implement a very fast algorithm 
that can ``un-gerrymander'' (with respect to the efficiency gap measure) 
the gerrymandered US house districts of three US states.

\captionsetup{width=\textwidth}
\begin{table}[h]
\begin{tabular}{l  |    c c  |  c c  |    c}
\toprule
  &
       \multicolumn{2}{c|}{\textbf{Vote share}}	& 
       \multicolumn{2}{c|}{\textbf{Number of Seats}}	& 
	        \multicolumn{1}{c}{\textbf{Normalized efficiency gap}}
\\
  &
	     \textbf{Democrats} & \textbf{GOP} &
	       \textbf{Democrats} & \textbf{GOP} &
	        \multicolumn{1}{c}{(current)}
\\
  &
	         $\pmb{\dfrac{\aalpha(\cP)}{\eeta(\cP)}}$ &  $\pmb{\dfrac{\bbeta(\cP)}{\eeta(\cP)}}$ &
	        & &
	        \multicolumn{1}{c}{$\pmb{\mathsf{Effgap}_\kappa(\cP,\dots\dots)/\eeta(\cP)}$}
\\
\midrule
Wisconsin & 
   $50.75\%$ & $49.25\%$ & $3$ & $5$ & $14.76 \%$
\\
\midrule
Texas & 
   $43.65\%$ & $56.35\%$ & $12$ & $24$ & $4.09 \%$ 
\\
\midrule
Virginia & 
      $51.96\%$ & $48.04\%$ & $4$ & $7$ & $22.25 \%$ 
\\
\midrule
Pennsylvania & 
      $50.65\%$ & $49.35\%$ & $5$ & $13$ & $23.80 \%$ 
\\
\bottomrule
\end{tabular}
\caption{\label{ttt1}Summary statistics for $2012$ election data for election of the (federal) house of representatives for  
the states of Texas, Wisconsin, Virginia and Pennsylvania.}
\end{table}
\setlength{\tabcolsep}{6pt}

A closer look at the new district maps shown in \FI{WImap-alg1}--\ref{PAmap-alg1}
also reveal the following interesting insights: 
\begin{description}
\item[Seat gain vs.\ efficiency gap.]
Lowering the efficiency gap from $15\%$ to $3.75\%$ for the state of Wisconsin did \emph{not} 
affect the total seat allocation ($3$ democrats vs.\ $5$ republicans) between the two parties.
Indeed, this further reinforces the assertion in~\cite{sm15} that 
\emph{the efficiency gap and partisan symmetry are different concepts}, and thus fewer absolute
difference of wasted votes does not necessarily lead to seat gains for the loosing party.
\item[Compactness vs.\ efficiency gap.]
The new district maps for the state of Virginia reveals an interesting aspect. 
Our new district map have fewer districts that are oddly shaped
compared to the map used for the $2012$ 
election\footnote{Virginia is one of the most gerrymandered states in the country, both on the congressional and state levels, 
based on lack of compactness and contiguity of its districts. 
Virginia congressional districts are ranked the 5th worst in the country because counties and 
cities are broken into multiple pieces to create heavily partisan districts~\cite{web-compact}.},
even though minimizing wasted votes 
does not take into consideration shapes of districts. 
\item[How natural are gerrymandered districts?]
Since our algorithm applies a sequence of carefully chosen semi-random \emph{perturbations} 
to the original gerrymandered districts to drastically lower the absolute difference of wasted votes,
one can \emph{hypothesize} that the original gerrymandered districts are far from 
being a product of \emph{arbitrarily random} decisions. However, to reach a definitive conclusion
regarding this point, one would need to construct a suitable null model, which we do not have yet.
\end{description}

\begin{figure}[htbp]
\captionsetup{width=\textwidth}
\begin{tabular}{l     c c | c c  c    c c}
\toprule
  &
       \multicolumn{4}{c}{\textbf{Number of Seats}}	& 
					$\,\,\,\,\,\,\,\,\,\,\,\,\,\,\,\,$ &
	        \multicolumn{2}{c}{\hspace*{-0.2in}\textbf{Normalized efficiency gap}}
\\
\cmidrule{2-5}
  &
       \multicolumn{2}{c|}{\textbf{Original}}	& 
       \multicolumn{2}{c}{\textbf{New}}	& 
					$\,\,\,\,\,\,\,\,\,\,\,\,\,\,\,\,$ &
	     \multicolumn{2}{c}{$\pmb{\mathsf{Effgap}_\kappa(\cP,\dots\dots)/\eeta(\cP)}\,\,\,\,\,\,\,\,\,\,\,\,$}
\\
  &
	       \textbf{Democrats} & \textbf{GOP} &
	       \textbf{Democrats} & \textbf{GOP} &
					$\,\,\,\,\,\,\,\,\,\,\,\,\,\,\,\,$ &
          \multicolumn{1}{c}{\textbf{Original}}	& 
          \multicolumn{1}{c}{\textbf{New}}
\\
\midrule
Wisconsin & $3$ & $5$ & $3$ & $5$ & & $14.76 \%$ & $3.80 \%$ 
\\
\midrule
Texas & $12$ & $24$ & $12$ & $24$ & & $4.09 \%$ & $3.33 \%$
\\
\midrule
Virginia & $3$ & $8$ & $5$ & $6$ & & $22.25 \%$ & $3.61 \%$
\\
\midrule
Pennsylvania & $5$ & $13$ & $6$ & $12$ & & $23.80 \%$ & $8.64 \%$
\\
\bottomrule
\end{tabular}
\captionof{table}{\label{t2}Redistricting results obtained by running the algorithm in \FI{alg1}
for the states of Texas, Wisconsin, Virginia and Pennsylvania in comparison to the $2012$ district plans.}
\end{figure}

\begin{figure}[p]
\centerline{\includegraphics[scale=0.95]{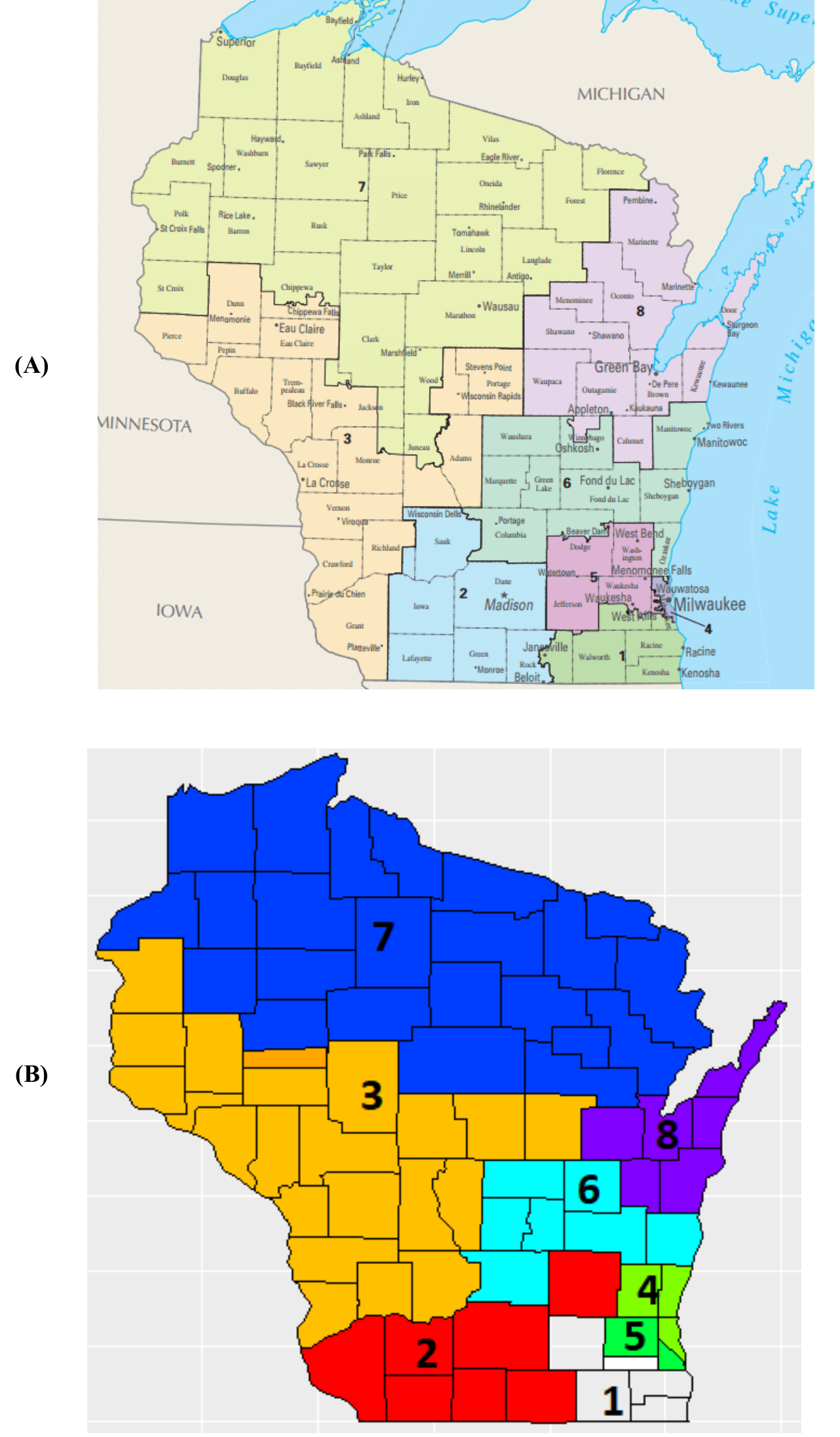}}
\caption{The district maps of Wisconsin: 
\textbf{(A)} original~\cite{web-map-wi}
and \textbf{(B)} after applying our local search algorithm in \FI{alg1}.}
\label{WImap-alg1}
\end{figure}

\begin{figure}[p]
\centerline{\includegraphics[scale=1]{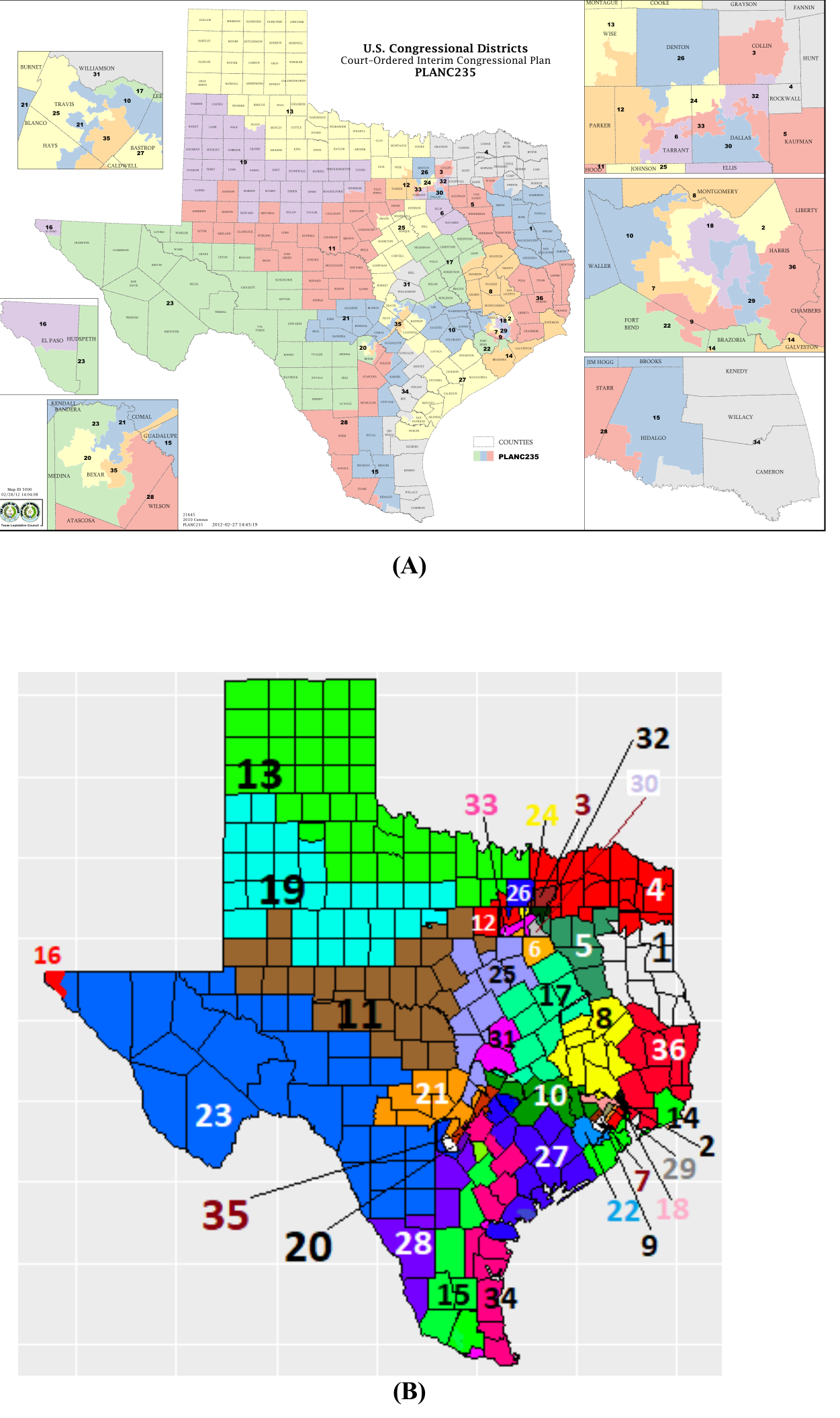}}
\caption{The district maps of Texas: 
\textbf{({A})} original~\cite{web-map-tx}
and \textbf{({B})} after applying our local search algorithm in \FI{alg1}.}
\label{TXmap-alg1}
\end{figure}

\begin{figure}[p]
\centerline{\includegraphics[scale=1]{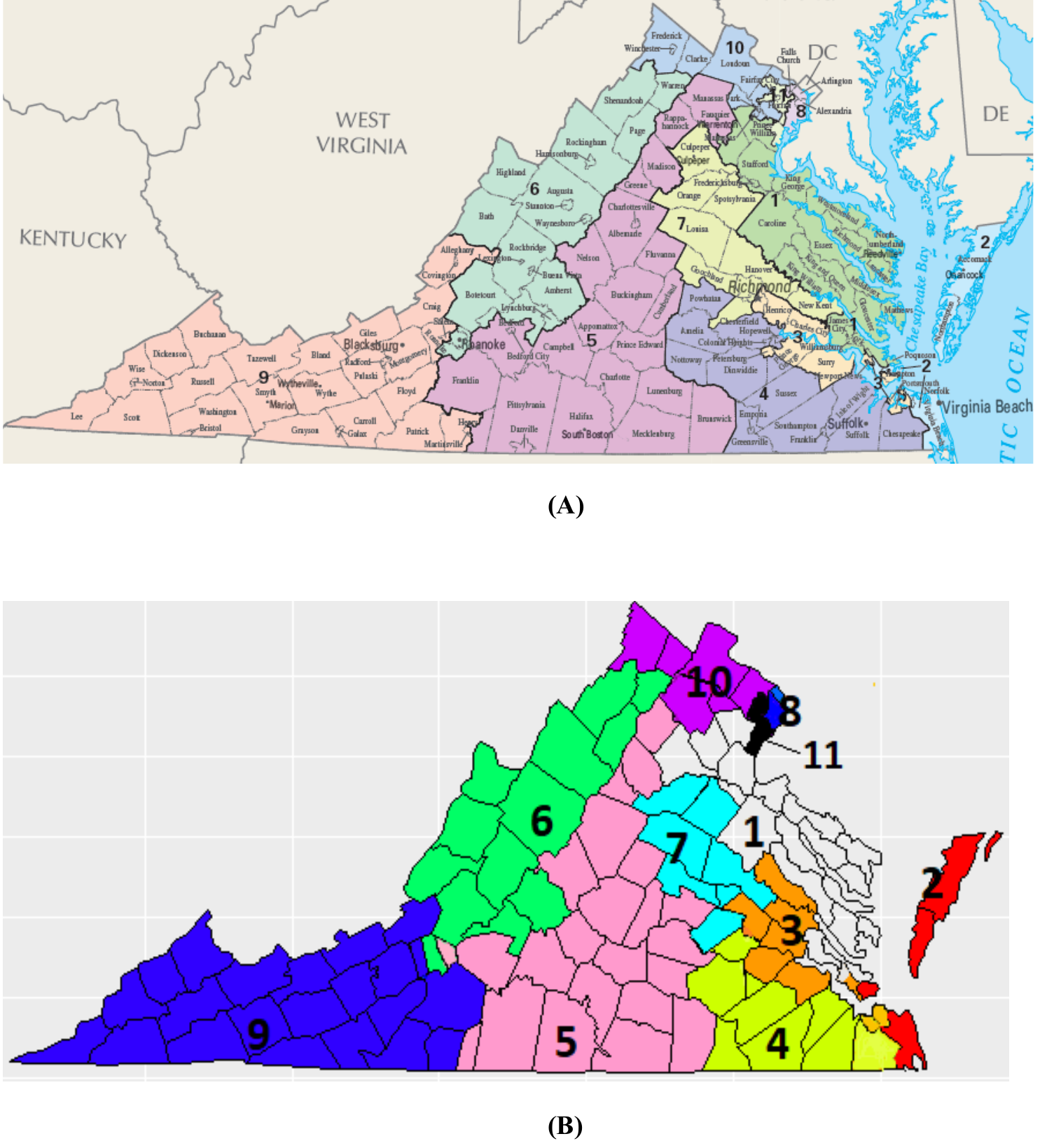}}
\caption{The district maps of Virginia: 
\textbf{({A})} original~\cite{web-map-va}
and \textbf{({B})} after applying our local search algorithm in \FI{alg1}.}
\label{VAmap-alg1}
\end{figure}

\begin{figure}[p]
\centerline{\includegraphics[scale=1]{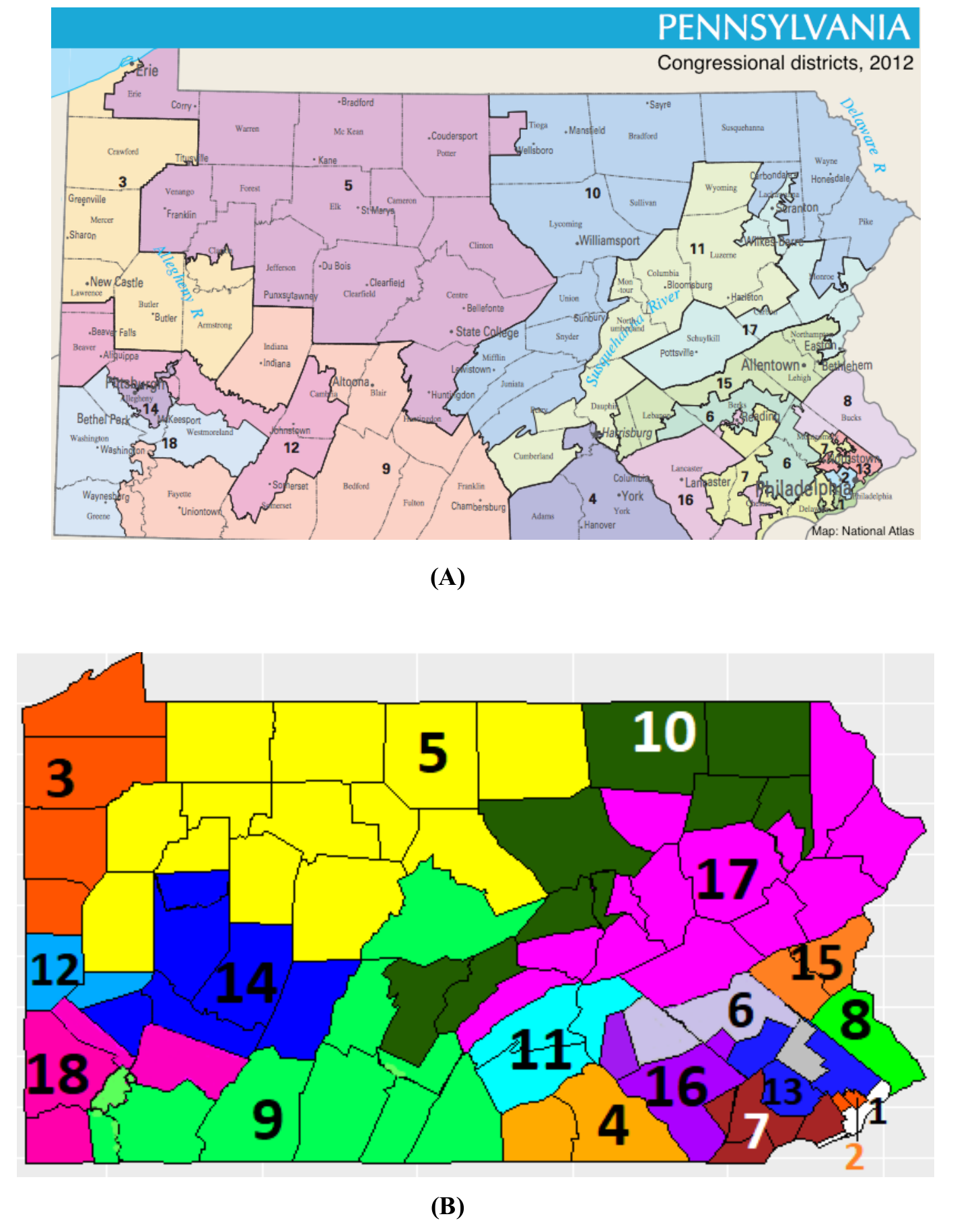}}
\caption{The district maps of Pennsylvania: 
\textbf{({A})} original~\cite{web-map-pa}
and \textbf{({B})} after applying our local search algorithm in \FI{alg1}.
}
\label{PAmap-alg1}
\end{figure}

\clearpage

\section{Conclusion and future research}

In this article we have performed algorithmic analysis of the recently introduced efficiency gap measure 
for gerrymandering both from a theoretical (computational complexity) as well as a practical (software 
development and testing on real data) point of view. 
The main objective of the paper was to provide a scientific analysis of the efficiency gap measure and 
to 
provide a crucial supporting hand to remove partisan gerrymandering 
should the US courts decide to recognize efficiency gap as at least a partially valid measure of 
gerrymandering. Of course, final words on resolving gerrymandering is up to the US judicial systems. 
%



%
\end{document}